\pdfoutput=1

\documentclass[11pt,a4paper]{article}
\usepackage{jheppub}
\usepackage{mciteplus}
\usepackage{hyperref}
\usepackage{amssymb}
\usepackage{amsthm}
\usepackage{amstext}
\usepackage{amsmath}
\usepackage{amsfonts}
\usepackage{dsfont}
\usepackage{cancel}
\usepackage{wrapfig}
\usepackage{graphicx}
\usepackage{pstricks}
\usepackage{color}

\usepackage{enumitem}

\usepackage{xparse}

\def\d{{\rm d}}

\newcommand{\compactspace}{ \hspace{0mm}}
\newcommand{\cutkernel}{\frac{\tilde P^\frac{D-L-E-1}{2}}{z_{r_1}
    \ldots  z_{r_{k-c}}z_{r_{k-c+1}}^{\nu_{r_{k-c+1}}} \ldots z_{r_{m-c}}^{\nu_{r_{m-c}}}}}

\NewDocumentCommand\Tensor{mmggg}{
#1 \otimes #2 \IfNoValueTF{#3}{}{\otimes #3} \IfNoValueTF{#4}{}{\otimes #4} \IfNoValueTF{#5}{}{\otimes #5}
}

\newtheorem{theorem}{Theorem}

\newtheorem{algorithm}[theorem]{Algorithm}

\newtheorem{definition}[theorem]{Definition}

\newtheorem{lemma}[theorem]{Lemma}

\newtheorem{remark}[theorem]{Remark}

\numberwithin{equation}{section}

\allowdisplaybreaks

\title{Complete integration-by-parts reductions of the non-planar
  hexagon-box via module intersections}
\author[a]{Janko B{\"o}hm,}
\author[b]{Alessandro Georgoudis,}
\author[c]{Kasper J. Larsen,}
\author[a]{Hans Sch{\"o}nemann,}
\author[d,e,f]{Yang~Zhang}
\affiliation[a]{Department of Mathematics, University of Kaiserslautern, 67663 Kaiserslautern, Germany}
\affiliation[b]{Department of Physics and Astronomy, Uppsala University, SE-75108 Uppsala, Sweden}
\affiliation[c]{School of Physics and Astronomy, University of Southampton, Highfield, Southampton, SO17 1BJ, United Kingdom}
\affiliation[d]{PRISMA Cluster of Excellence, Johannes Gutenberg
  University, 55128 Mainz, Germany}
\affiliation[e]{Institute for Theoretical Physics, ETH Z\"urich, CH
  8093 Z\"urich, Switzerland}
\affiliation[f]{Kavli Institute for Theoretical Physics,
University of California, Santa Barbara, CA 93106, USA}
\emailAdd{boehm@mathematik.uni-kl.de}
\emailAdd{Alessandro.Georgoudis@physics.uu.se}
\emailAdd{Kasper.Larsen@soton.ac.uk}
\emailAdd{hannes@mathematik.uni-kl.de}
\emailAdd{zhang@uni-mainz.de}

\abstract{We present the powerful module-intersection integration-by-parts (IBP)
  method, suitable for multi-loop and multi-scale Feynman
  integral reduction. Utilizing modern computational algebraic
  geometry techniques, this new method successfully trims traditional IBP systems dramatically to
  much simpler integral-relation systems on unitarity cuts. We
  demonstrate the power of this method by explicitly carrying out
  the complete analytic reduction of two-loop five-point non-planar
  hexagon-box integrals, with degree-four numerators, to a basis of
  $73$ master integrals.}


\begin{document}

\begin{flushright}\begin{tabular}{r}
MITP/18-034 \\
UUITP-16/18
\end{tabular}\end{flushright}
\vspace{-14.1mm}
\maketitle

\clearpage

\section{Introduction}

High precision in the theoretical predictions for cross sections is
necessary for the analysis of the Large
Hadron Collider (LHC) Run II experiments. For many processes, the
next-to-next-to-leading-order (NNLO) contributions are needed for
precision-level predictions. NNLO in general requires the computation of a
large number of two-loop (or, in some cases, higher-loop) Feynman integrals,
which is often a bottleneck problem for particle physics.

Integration-by-parts (IBP) reduction is a key tool to reduce the large
number of multi-loop Feynman integrals to a small integral basis of
so-called master integrals. Schematically, for an $L$-loop integral in
dimensional regularization with the inverse propagators $D_1, \ldots, D_m$, we have
\begin{equation}
0=\int \frac{\d^D \ell_1}{\mathrm{i} \pi^{D/2}} \ldots \frac{\d^D \ell_L}{\mathrm{i} \pi^{D/2}}
\sum_{j=1}^L \frac{\partial}{\partial \ell_j^\mu}
\frac{v_j^\mu \hspace{0.5mm} }{D_1^{\nu_1} \cdots D_m^{\nu_m}}
\,, \label{eq:IBP_schematic}
\end{equation}
where the $v_j^\mu$ are vectors constructed from the external momenta and the loop
momenta. Expanding the action of the derivative in eq.~\eqref{eq:IBP_schematic} leads to an
IBP identity, which is a linear relation between multi-loop Feynman integrals. After collecting
a sufficient set of IBP identities, one can carry out reduced row reduction in a
specific integral ordering to express a target set of integrals as a linear combination of the master integrals.

The row reduction of IBP identities can be achieved by
the Laporta algorithm~\cite{Laporta:2000dc,Laporta:2001dd}. There are several publicly available implementations
of IBP reductions: AIR~\cite{Anastasiou:2004vj},
FIRE~\cite{Smirnov:2008iw,Smirnov:2014hma},
Reduze~\cite{Studerus:2009ye,vonManteuffel:2012np},
LiteRed~\cite{Lee:2012cn}, Kira~\cite{Maierhoefer:2017hyi},
as well as various private implementations. In the standard Laporta algorithm and
many of its variants, integrals with doubled propagators appear in the intermediate
steps and also in the master integral basis. To trim the linear
system of IBP identities for row reduction, one may impose the condition
that no integrals with doubled propagators appear in the IBP identities \cite{Gluza:2010ws}. This condition can be
solved by the syzygy computation \cite{Gluza:2010ws} or linear algebra \cite{Schabinger:2011dz}.
In refs.~\cite{Ita:2015tya, Larsen:2015ped}, methods for deriving IBP reductions on
generalized-unitarity cuts was introduced. IBP reduction can also be
sped up by the finite-field sampling method~\cite{vonManteuffel:2014ixa}.
For several multiple-scale Feynman integrals, inspired choices of IBP generating
vectors $v_j^\mu$, leading to simple IBP identities, can be obtained from dual conformal
symmetry \cite{Bern:2017gdk}. In a recent development, IBP identities
which target integrals with arbitrary-degree numerators
\cite{Kosower:2018obg} were introduced. An initial (and simpler) step
to carrying out the IBP reductions is that of determining the integral basis itself.
This can be done by the packages {\sc Mint}
\cite{Lee:2013hzt} or {\sc Azurite} \cite{Georgoudis:2016wff}.
The dimension of the integral basis can be efficiently determined by the computation of
an Euler characteristic which naturally arises in the
D-module theory of polynomial annihilators \cite{Bitoun:2017nre}.

IBP reductions, combined with generalized-unitarity methods
\cite{Britto:2004nc,Britto:2005fq,Kosower:2011ty,Johansson:2012zv} or integrand
reduction methods \cite{Ossola:2006us, Badger:2012dp, Zhang:2012ce, Mastrolia:2012an},
have led to successful computations of many complicated
multi-loop integrands. Furthermore, IBP reductions are also important for
setting up the differential equations for Feynman integrals~\cite{Kotikov:1990kg,Kotikov:1991pm,Bern:1993kr,Remiddi:1997ny,Gehrmann:1999as,Henn:2013pwa,%
Papadopoulos:2014lla,Lee:2014ioa,Ablinger:2015tua,Papadopoulos:2015jft,Liu:2017jxz}, which has proven
a highly efficient method for evaluating Feynman integrals. Recent developments in
multi-loop integral reduction and the differential equation method have produced very
impressive results for the two-loop amplitudes of $2\to 3$ scattering processes
with the
all-plus-helicity configuration
\cite{Badger:2013gxa,Badger:2015lda,Gehrmann:2015bfy} and the generic-helicity configuration
\cite{Badger:2017jhb,Abreu:2017hqn,Boels:2018nrr}.
IBP reductions are also crucial for deriving dimension recursion relations
\cite{Tarasov:1996br,Lee:2009dh,Lee:2010wea}.

For multi-loop, high-multiplicity or multi-scale amplitude computations, the IBP reductions
frequently become the bottleneck which requires extensive computing resources and
time. Schematically, the complexity of IBP reductions mainly originate from the following facts,
\begin{enumerate}
\item Multi-loop IBP identities involve many contributing integrals, and the
  complexity of row reduction grows quickly with the number of integrals.
\item There are many external invariants (i.e., Mandelstam variables and mass parameters),
  and the algebraic computation of polynomials or rational functions in
  these parameters requires a substantial amount of CPU time and RAM.
\end{enumerate}

In this paper, utilizing our new powerful module-intersection IBP reduction
method, we present solutions to the above problems:
\begin{enumerate}
\item With the module-intersection computation \cite{Zhang:2016kfo}, we dramatically reduce
  the number of relevant IBP identities and integrals involved by
  imposing the no-doubled-propagator condition and applying unitarity cuts.
\item With the classical technique of treating parameters as new variables
  and imposing a special monomial ordering, originating from primary
  decomposition algorithms \cite{GTZ} in commutative algebra, we can efficiently compute analytic IBP reductions
  with many parameters.
\end{enumerate}
In this paper, we demonstrate our method by treating the cutting-edge example
of the analytic IBP reduction of two-loop five-point non-planar
hexagon-box integrals, with all degree-$1,2,3,4$ numerators, to the basis of
$73$ master integrals. This computation was carried out using private implementations
in {\sc Singular} \cite{Singular} and {\sc Mathematica} of the ideas
presented in this paper. Our result of fully analytic IBP reductions can be
downloaded from,
\begin{quote}
  \url{https://github.com/yzhphy/hexagonbox_reduction/releases/download/1.0.0/hexagon_box_degree_4_Final.zip} \,.
\end{quote}

This paper is organized as follows. In section~\ref{section:IBP_module_intersection}
we present our module-intersection method for trimming IBP systems and describe
a simple way of obtaining the individual modules needed for our
method. In section~\ref{section_intersection}, the central part of our
paper, we present a highly efficient new method for computing these
intersections. We then explain the technical details of performing
row reduction for the trimmed IBP systems from module intersections
in section~\ref{section_REF}. In section~\ref{section_hexagon_box}, we
present our example, the module-intersection IBP reduction of hexagon-box
integrals, in details. The conclusions and outlook are given in section~\ref{section_conclusion}.

\section{Simplification of IBP systems by module-intersection method}
\label{section:IBP_module_intersection}
In this section we present the details on trimming IBP systems by
requiring the absence of integrals with doubled propagators \cite{Gluza:2010ws},
and by furthermore applying unitarity cuts. To demonstrate our method clearly,
we first explain how to efficiently impose the no-doubled-propagator condition
without applying unitarity cuts. Then we show that the same approach can be applied on
unitarity cuts. In both cases, the algebraic constraints for trimming
IBP systems are reformulated as the problem of computing the intersection of two modules
over a polynomial ring \cite{Zhang:2016kfo}. At the end of this section,
the algorithm for computing each individual module is introduced,
while the intersection algorithm, the heart of this method,
will be introduced in section~\ref{section_intersection}.

\subsection{Module intersection method without unitarity cuts}
\label{section:Module_intersection_no_cut}
The objects under consideration are the multi-loop Feynman integrals,
\begin{equation}
  I(\nu_1,\ldots, \nu_m) = \int \frac{\d^D l_1}{i \pi^{D/2}} \ldots
  \frac{\d^D l_L}{i \pi^{D/2}}  \frac{1}{D_1^{\nu_1} \cdots
    D_m^{\nu_m}}\,.
\label{Feynman_integral_1}
\end{equation}
Here $\nu_i \in \mathbb{Z}$, $l_1, \ldots, l_L$ denote the loop momenta, and $p_1, \ldots, p_E$ the {\it independent} external momenta.
By the procedure of integrand reduction, we can set
$m=L(L+1)/2+LE$. The $D_i$ denote inverse propagators.

To study IBP reduction of integrals of the form (\ref{Feynman_integral_1}), we focus on the family of Feynman integrals associated with a particular Feynman diagram with $k$ propagators (where $k\leq m$) and all its daughter diagrams, obtained by pinching propagators. Without loss of generality, the inverse denominators of this Feynman diagram can be labeled as $D_1, \ldots, D_k$. Therefore, the family of integrals is,
\begin{equation}
  \label{eq:3}
 \mathcal F=\{ I(\nu_1,\ldots, \nu_m) \hspace{1mm}| \quad \nu_{j} \leq 0 \text{ if } j>k \} \,.
\end{equation}

Traditionally, IBP reduction is carried out within such a family.
However, since the majority of the Feynman integrals that contribute to a quantity in perturbative QFT
are integrals without doubled propagators, it is natural to consider the subfamily
\begin{equation}
  \label{eq:3}
 \mathcal F_\mathrm{ndp}= \{I(\nu_1,\ldots, \nu_m) \hspace{1mm}| \quad  \nu_{j} \leq 1 \text{ if } j\leq k,\quad \nu_{j} \leq 0 \text{ if } j>k  \}
\end{equation}
and the IBP relations for integrals in this subfamily \cite{Gluza:2010ws}.

We find that it is convenient to use the Baikov representation
\cite{Baikov:1996rk} of Feynman integrals for their integration-by-parts (IBP) reduction for several reasons: 1) the integrand reduction is manifest in this representation, 2) it is easy to apply unitarity cuts, 3) most importantly, it is surprisingly simple to trim IBP systems {\it analytically} in this representation. Here we briefly review the Baikov representation.

We collect the external and internal momenta as,
\begin{equation}
V = ( v_1, \ldots, v_{E+L}) = ( p_1, \ldots, p_E, \ell_1, \ldots, \ell_L ) \,.
\label{eq:list_of_momenta}
\end{equation}
The Gram matrix $S$ of these vectors is,
\begin{equation}
S =
\left(
 \begin{array}{ccc|ccc}
  x_{1,1}  &  \compactspace \cdots \compactspace  &  \compactspace x_{1,E}    &  x_{1,E+1}    &  \compactspace \cdots \compactspace  &  \compactspace x_{1,E+L} \\
  \vdots   &  \compactspace \ddots \compactspace  &  \compactspace \vdots     &  \vdots       &  \compactspace \ddots \compactspace  &  \compactspace \vdots  \\
  x_{E,1}  &  \compactspace \cdots \compactspace  &  \compactspace x_{E,E}    &  x_{E,E+1}    &  \compactspace \cdots \compactspace  &  \compactspace x_{E,E+L} \\
\hline
 x_{E+1,1} &  \compactspace \cdots \compactspace  &  \compactspace x_{E+1,E}  &  x_{E+1,E+1}  &  \compactspace \cdots \compactspace  &  \compactspace x_{E+1,E+L} \\
 \vdots    &  \compactspace \ddots \compactspace  &  \compactspace \vdots     &  \vdots       &  \compactspace \ddots \compactspace  &  \compactspace \vdots  \\
 x_{E+L,1} &  \compactspace \cdots \compactspace  &  \compactspace x_{E+L,E}  &  x_{E+L,E+1}  &  \compactspace \cdots \compactspace  &  \compactspace x_{E+L,E+L}
\end{array} \right) \,,
\label{eq:extended_Gram_matrix}
\end{equation}
where the elements are defined as $x_{i,j} = v_i \cdot v_j$. The
upper-left $E\times E$ block is the Gram matrix of the external momenta,
which is denoted as $G$. Defining
$z_i\equiv D_i$ and integrating out solid-angle directions, the Feynman
integral (\ref{Feynman_integral_1}) takes the following form in Baikov representation,
\begin{align}
I(\nu_1,\ldots, \nu_m)  \hspace{0.5mm}&=\hspace{0.5mm}  C_E^L \hspace{0.7mm} U^\frac{E-D+1}{2} \hspace{-1mm}
\int \d z_1 \cdots \d z_m P^\frac{D-L-E-1}{2} \frac{1}{z_1^{\nu_1} \cdots z_m^{\nu_m}}  \,,
\label{eq:Baikov_representation}
\end{align}
where $P\equiv \det S$, $U\equiv \det G$, and the factor $C_E^L$ originates from
the solid-angle integration and the Jacobian for the transformation $x_{i,j} \to z$. For
the derivation of IBP identities, $U$ and $C_E^L$ are irrelevant, so we
may ignore them in the following discussion. Note that in this
representation, the inverse denominators $D_i$ become free variables,
and so the integrand reduction can be done automatically.

An IBP identity in this representation reads,
\begin{align}
  0&=\int \d z_1 \cdots \d z_m \sum_{i=1}^m\frac{\partial}{\partial
     z_i} \bigg(a_i P^\frac{D-L-E-1}{2} \frac{1}{z_1^{\nu_1} \cdots
     z_m^{\nu_m}}  \bigg) \nonumber \\
&=\int \d z_1 \cdots \d z_m \sum_{i=1}^m\bigg(\frac{\partial a_i}{\partial
     z_i} + \frac{D-L-E-1}{2 P} a_i \frac{\partial P}{\partial
     z_i} - \frac{\nu_i a_i}{z_i}\bigg) P^\frac{D-L-E-1}{2} \frac{1}{z_1^{\nu_1} \cdots
     z_m^{\nu_m}} \,.
\label{IBP_Baikov}
\end{align}
Here the $a_i$ denote polynomials in the ring $\mathbb A=\mathbb Q(\mathbf
s)[z_1,\ldots, z_m]$. 
 (We use $\mathbf s$ to represent the
independent Mandelstam variables and mass parameters collectively.)
We remark that in the Baikov representation, the Baikov polynomial $P$
vanishes on the boundary of the integration domain, and hence there is no
surface term in the IBP identity. Note that the terms with the pole
$1/P$ appearing inside the parenthesis in eq.~(\ref{IBP_Baikov})
correspond to dimension-shifted integrals, and as such are not
favorable in deriving simple IBP relations. To avoid such poles, we may impose the
following constraints on the
$a_i$ \cite{Lee:2014tja, Ita:2015tya, Larsen:2015ped},
\begin{equation}
  \bigg(\sum_{i=1}^m a_i \frac{\partial P}{\partial z_i} \bigg)+ b
  P=0\,,
\label{module1}
\end{equation}
where $b$ is also required to be a polynomial in $\mathbb A$. This
constraint is known in computational commutative algebra as a
``syzygy'' equation \cite{Gluza:2010ws}. In the following discussion, we only focus on
the polynomials $a_i$, since once they are known it is straightforward
to recover the polynomial $b$. The solutions of eq.~\eqref{module1}, taking the form,
\begin{equation}
  (a_1,\ldots, a_m)
\end{equation}
form a sub-module of the polynomial module $\mathbb A^{m}$, which we denote $M_1$
in the following.

Furthermore, to trim the IBP system, it is possible to work with
integrals in $\mathcal F_\mathrm{ndp}$, i.e.~integrals without doubled propagators \cite{Gluza:2010ws}. Note
from the second line of eq.~\eqref{IBP_Baikov} that, even if the integral inside the differential
operator has $\nu_i\leq 1$, the derivative will produce integrals
with doubled propagators. This can also be prevented by a suitable
choice of the $a_i$. For example, if we consider the integral of the parent
diagram, $\nu_1=\cdots=\nu_k=1$, $\nu_{k+1}\leq 0,\ldots, \nu_m\leq 0$
inside eq.~\eqref{IBP_Baikov}, we can avoid doubled propagators
 by requiring that $a_i$ is divisible by $z_i$,
\begin{equation}
  a_i = b_i z_i \,,\quad i=1,\ldots, k \,.
\label{module2}
\end{equation}
Such $(a_1,\ldots, a_m)$ also form a sub-module of $\mathbb A^m$, which we denote
$M_2$. We need to solve eqs.~\eqref{module1} and \eqref{module2}
simultaneously to find IBP relations which involve neither dimension shifts
nor doubled propagators in Baikov representation.

The strategy in ref.~\cite{Larsen:2015ped} to solve
these conditions is to replace $a_i$ by $b_i$ in eq.~\eqref{module1} and
then solve for $(b_1, \ldots, b_k, a_{k+1},\ldots, a_m, b)$ as a single syzygy equation,
employing Schreyer's theorem. However, this approach, although it works well for
simple two-loop four-point integrals, becomes less practical for more
complicated kinematics. It was suggested in the reference \cite{Zhang:2016kfo}
that a better strategy is to determine the generators of $M_1$ and
$M_2$ individually, and then calculate the {\it module intersection},
\begin{equation}
  \label{eq:6}
  M_1\cap M_2 \,,
\end{equation}
whose generators are solutions of eqs.~\eqref{module1} and
\eqref{module2}. Geometrically, elements in  $M_1\cap M_2$ are the
polynomial tangent vectors of the reducible hypersurface
\cite{Hauser1993,Zhang:2016kfo}
\begin{equation}
  \label{eq:12}
  z_1 \ldots z_k P=0\,.
\end{equation}

In subsection
\ref{subsection_individual_module}, we will see that it takes no effort to find
the generators of $M_1$ and $M_2$. In section \ref{section_intersection}
we present a highly efficient algorithm for computing the intersection $M_1\cap M_2$.

Once eqs.~\eqref{module1} and \eqref{module2} are solved, we obtain the
simplified IBP identities without doubled propagators. They take the following form,
\begin{equation}
  \label{eq:2}
  0=\int \d z_1 \cdots \d z_m \bigg(\sum_{i=1}^m \frac{\partial a_i}{\partial
     z_i} - \sum_{i=1}^k  b_i  -\sum_{i=k+1}^m \frac{\nu_i a_i}{z_i} - \frac{D-L-E-1}{2} b \bigg) \frac{P^\frac{D-L-E-1}{2}}{z_1\ldots
    z_k  z_{k+1}^{\nu_{k+1}} \ldots z_m^{\nu_m}}\,.
\end{equation}
Once the generators of $M_1\cap M_2$ are obtained, we multiply them by monomials
in the $z_i$ in order to get a spanning set of IBP identities for the reduction targets.
Alternatively, it is possible to apply D-module theory to get the IBP reductions directly from
the generators of $M_1\cap M_2$. We leave this direction for future research.

\subsection{Module intersection method with cuts}
\label{subsection_module_intersection_on_cuts}

In practice, for Feynman diagrams with high multiplicity or high loop
order, instead of working with the Feynman integrals directly, it is
more convenient to apply unitarity cuts \cite{Ita:2015tya,Larsen:2015ped}.
In this section we show how to apply the module intersection method
in combination with unitarity cuts, thus simplifying the construction
and subsequent Gaussian elimination of the IBP identities.
We follow the notation of ref.~\cite{Boehm:2017wjc}.

Consider the $c$-fold cut of eq.~\eqref{Feynman_integral} with
$c\leq k$. Let $\mathcal{S}_\mathrm{cut}$, $\mathcal{S}_\mathrm{uncut}$
and $\mathcal{S}_\mathrm{ISP}$ denote the sets of indices of cut propagators,
uncut propagators and irreducible scalar products (ISP) respectively. Explicitly,
\begin{align}
\mathcal{S}_\mathrm{cut}    =  \{ \zeta_1, \ldots, \zeta_c \}\,, \quad
  \mathcal{S}_\mathrm{uncut} =  \{ r_1, \ldots, r_{k-c}\} \,, \quad
\mathcal{S}_\mathrm{ISP}    =  \{ r_{k-c+1}, \ldots, r_{m-c}\}  \,.
\end{align}
In Baikov represention, the $c$-fold cut of integrals with $\nu_1 = \cdots =\nu_k=1$
takes a simple form,
\begin{align}
I(\nu_1,\ldots, \nu_m)  \hspace{0.5mm}&\propto
\int \d z_{r_1} \cdots \d z_{r_{m-c}} \cutkernel \,,
\label{eq:Baikov_representation_cut}
\end{align}
where,
\begin{equation}
  \label{eq:4}
  \tilde P \hspace{0.8mm}=\hspace{0.8mm} P |_{z_{\zeta_1}\to 0,\ldots, z_{\zeta_c}\to 0}\,.
\end{equation}
We can derive IBP identities for the integrals on the cut $\mathcal{S}_\mathrm{cut}$,
\begin{align}
  0&=\int \d z_{r_1} \ldots \d z_{r_{m-c}} \sum_{i=1}^{m-c}\frac{\partial}{\partial
     z_{r_i}} \bigg(\tilde a_{r_i} \cutkernel \bigg) \nonumber \\
&=\int \d z_{r_1} \ldots \d z_{r_{m-c}}
  \bigg(\sum_{i=1}^{m-c}\frac{\partial \tilde a_{r_i}}{\partial
     z_{r_i}} + \frac{D-L-E-1}{2 \tilde P} \sum_{i=1}^{m-c} \tilde a_{r_i}
  \frac{\partial \tilde P}{\partial
     z_{r_i}} - \sum_{i=1}^{k-c} \frac{\tilde a_{r_i}}{z_{r_i}}
  -\sum_{i=k-c+1}^{m-c}\frac{\nu_{r_i} \tilde a_{r_i}}{z_{r_i}} \bigg) \nonumber \\
& \hspace{12mm} \times \cutkernel \,,
\label{IBP_Baikov_cut}
\end{align}
where $\tilde a_{r_i}$, $i=1, \ldots, m-c$ are polynomials in $\tilde{ \mathbb A}=\mathbb Q(\mathbf
s)[z_{r_1},\ldots z_{r_{m-c}}]$. Once again, to derive simple IBP relations
without dimension shifts or doubled propagators
\cite{Ita:2015tya,Larsen:2015ped}, we may impose the conditions,
\begin{align}
  \Big(\sum_{i=1}^{m-c} \tilde a_{r_i} \frac{\partial \tilde
  P}{\partial z_{r_i}} \Big) + \tilde b \tilde P &=0 \,, \label{cut_constraint_1_original}\\
\tilde a_{r_i} &= \tilde b_{r_i} z_{r_i},\quad i=1,\ldots, k-c \,.
\label{cut_constraint_2_original}
\end{align}
Again, the module intersection method provides an efficient tool to solve
these constraints simultaneously. However, for the purpose of automation, it
is useful to make use of the following slight reformulation. Define
\begin{align}
  \label{reformulation}
  \tilde a_{\zeta_i}\equiv 0\,, \quad i=1, \ldots, c\,,
\end{align}
Then eqs.~\eqref{cut_constraint_1_original} and \eqref{cut_constraint_2_original}
are formally recast as,
\begin{align}
  \Big(\sum_{i=1}^{m} \tilde a_i \frac{\partial \tilde
  P}{\partial z_i} \Big) + \tilde b \tilde P &=0\,, \label{cut_constraint_1}\\
\tilde a_{r_i} &= \tilde b_{r_i} z_{r_i}\,,\quad i=1,\ldots, k-c \,.
\label{cut_constraint_2}
\end{align}
By comparing these equations and their counterparts without applied cuts,
we observe the following shortcut. Recall that the modules $M_1$ and $M_2$ defined
in the previous subsection are the solution sets for eqs.~\eqref{module1}
and \eqref{module2} respectively. We define
\begin{equation}
\tilde M_1 =M_1 |_{z_{\zeta_1}\to 0,\ldots ,z_{\zeta_c}\to 0}\,,\quad \tilde M_2 =M_2 |_{z_{\zeta_1}\to 0,\ldots ,z_{\zeta_c}\to 0} \,,
  \label{module_cut}
  \end{equation}
whereby the $c$-fold cut on the elements of $M_1$ and $M_2$ has been imposed.
Then clearly,
\begin{equation}
  \label{module_intersection_cut}
  \tilde M_1 \cap \tilde M_2
\end{equation}
solve the equations \eqref{cut_constraint_1} and
\eqref{cut_constraint_2} simultaneously. Note that any element in
$(q_1,\ldots, q_m)\in M_2$ has the property $q_{\zeta_i}= h_{\zeta_i}
z_{\zeta_i}$, $i=1,\ldots, c$. After imposing the cut we have,
\begin{equation}
  \label{eq:5}
  q_{\zeta_i} \big|_{z_{\zeta_1}\to 0,\ldots ,z_{\zeta_c}\to 0} =0\,,
\end{equation}
which is consistent with the requirement \eqref{reformulation}. This
automates the computation: to obtain $\tilde M_1$ and $\tilde M_2$,
we first compute $M_1$ and $M_2$, and then simply apply the rules for the various cuts.
An algorithm for computation of the intersection
will be introduced in section~\ref{section_intersection}.

Once the constraints \eqref{cut_constraint_1} and
\eqref{cut_constraint_2} are solved, the IBP identities on the cut take the
simple form,
\begin{gather}
  0 =\int \d z_{r_1} \cdots \d z_{r_{m-c}}\
  \bigg(\sum_{i=1}^{m-c}\frac{\partial \tilde a_{r_i}}{\partial
     z_{r_i}} - \frac{D-L-E-1}{2 } \tilde b- \sum_{i=1}^{k-c}  \tilde b_{r_i}
  -\sum_{i=k-c+1}^{m-c}\frac{\nu_{r_i} \tilde a_{r_i}}{z_{r_i}} \bigg)
  \nonumber \\
 \times \cutkernel\,.
\label{IBP_cut}
\end{gather}
As for the uncut case, once the generators of $\tilde M_1\cap
\tilde M_2$ are obtained, we multiply them by monomials in the $z_{r_i}$,
$i=1,\ldots, m-c$ to get a spanning set of IBP identities.

The cuts necessary for reconstructing the complete IBP identities can be
determined from a list of master integrals \cite{Larsen:2015ped,Georgoudis:2016wff}:
they are the maximal cuts of ``uncollapsible'' master integrals
(master integrals which cannot be obtained from other integrals in the
basis by adding propagators). In
practice, the list of master integrals can be quickly determined by the packages
{\sc Mint} \cite{Lee:2017oca}, {\sc Azurite}
\cite{Georgoudis:2016wff}. The total number of master integrals can
also be determined by the D-module theory method of ref.~\cite{Bitoun:2017nre}.

Upon applying Gaussian elimination to the IBP identities evaluated on each cut,
we can subsequently merge the obtained reductions to find the complete IBP reductions
without applied cuts.

\subsection{Algorithm for computing individual modules}
\label{subsection_individual_module}

The modules $M_1$ and $M_2$ defined in subsection \ref{section:Module_intersection_no_cut},
and hence $\tilde M_1$ and $\tilde M_2$ defined in the previous subsection,
can all be determined without effort.

The condition \eqref{module1} for $M_1$ is a syzygy equation for the
Baikov polynomial $P$ and its derivatives. Schreyer's theorem
\cite{MR1322960} guarantees that solutions for syzygy equations
can be obtained from Gr\"obner basis computations.
However, for the Baikov polynomial $P$ this is not necessary,
owing to the special structure of $P$. A convenient way to find
the solution, or equivalently the tangent vectors for the hypersurface $P=0$,
is to use the basic canonical IBP identities \cite{Ita:2015tya}.  Here alternatively, we use the Laplace
expansion method \cite{Boehm:2017wjc} \footnote{We thank Roman Lee for
introducing us to the Laplace expansion relations of symmetric matrices, also explained at his website
\url{http://mathsketches.blogspot.ru/2010/07/blog-post.html} (in Russian).} of the Gram
determinant $S$ \eqref{eq:extended_Gram_matrix} to determine $M_1$, since the results are
manifestly expressed in the $z$ variables.

Laplace expansion of $P=\det S$ yields,
\begin{equation}
\bigg(\sum_{k=1}^{E+L} (1{+}\delta_{i,k}) x_{j,k} \frac{\partial
  P}{\partial x_{i,k}}\bigg)- 2\delta_{i,j} P = 0 \,,
\label{eq:Laplace_expansion_of_Baikov_polynomial}
\end{equation}
where $E+1\leq i \leq E+L$ and $1\leq j \leq E+L$. These $L(E+L)$
relations are syzygy relations between $P$ and its derivatives in the $x_{i,j}$
variables. It is straightforward to convert them to solutions of
eq.~\eqref{module1},
\begin{equation}
\sum_{\alpha=1}^m (a_{i,j})_\alpha \frac{\partial P}{\partial z_\alpha} + b_{i,j} P = 0 \,,
\end{equation}
where $a_{i,j}$ and $b$ are obtained by applying the chain rule to
the expressions in eq.~\eqref{eq:Laplace_expansion_of_Baikov_polynomial}.
Explicitly, they are given by,
\begin{equation}
(a_{i,j})_\alpha =  \sum_{k=1}^{E+L} (1 + \delta_{i,k}) \frac{\partial z_\alpha}{\partial x_{i,k}} x_{j,k}
\hspace{5mm} \mathrm{and} \hspace{5mm}
b_{i,j}          = -2\delta_{i,j} \,.
\label{M1_generators}
\end{equation}
It is proven in ref.~\cite{Boehm:2017wjc} via J{\'o}zefiak complexes \cite{Joz78} 
that the $L(E+L)$ tuples of polynomials $(a_{i,j})_\alpha$ form 
a complete generating set of $M_1$.

We remark that,
\begin{enumerate}
  \item The generating set (\ref{M1_generators}) is at most linear in the $z_i$.

  \item The generating set (\ref{M1_generators}) is homogenous in the $z_i$
  {\it and} the Mandelstam variables/mass parameters, as can be inferred from dimensional analysis.
\end{enumerate}
The second property is crucial for our highly efficient algorithm for computing
intersections of modules, to be explained in section~\ref{section_intersection}.

The generating set for $M_2$ is trivial. There are $m$ generators,
\begin{align}
  z_1 \mathbf{e}_1 \,, \quad \ldots, \quad  z_k \mathbf{e}_k,\quad
  \mathbf{e}_{k+1}\,,\quad
    \ldots, \quad   \mathbf{e}_{m} \,.
\label{M2_generators}
\end{align}
Here $\mathbf{e}_j$ is the $j$-th $m$-dimensional unit vector.

The cut cases, $\tilde M_1$ and $\tilde M_2$ defined in
eq.~\eqref{module_cut}, can then be obtained from eqs.~\eqref{M1_generators} and
\eqref{M2_generators} by simply setting $z_{\zeta_i} \to 0$,
$i=1,\ldots, c$.

\section{Determining the module intersection}
\label{section_intersection}
In this section, we introduce a highly efficient algorithm for computing generators of the module intersections and an algorithm to trim the generating sets.

Given the two submodules $M_{1},M_{2}\subset A^{t}$ over the multivariate
polynomial ring $A=F[z_{1},\ldots,z_{n}]$ with coefficients in the
multivariate rational function field $F=\mathbb{Q}(c_{1},\ldots,c_{r})$, our
goal is to obtain a generating system of the module $M_{1}\cap M_{2}$ (which
is finitely generated since $A$ is Noetherian). We apply Gr\"{o}bner basis
techniques to address this problem. Our algorithms are implemented in the
computer algebra system \textsc{Singular}, which focuses on polynomial
computations with applications in commutative algebra and algebraic geometry
\cite{Singular}.

We first recall some terminology: Using the notation $z^{\alpha}=z_{1}%
^{\alpha_{1}}\cdot\ldots\cdot z_{n}^{\alpha_{n}}$ for the monomials in $A$, we
call $z^{\alpha}e_{i}$ with a unit basis vector $e_{i}\in A^{t}$ a
\emph{monomial} of $A^{t}$. An $F$-multiple of a monomial is called a
\emph{term}. A \emph{monomial ordering} on $A^{t}$ is an total ordering $>$ on
the set of monomials in $A^{t}$, which respects multiplication, that is,
$z^{\alpha}e_{i}>z^{\beta}e_{j}$ implies $z^{\alpha}z^{\gamma}e_{i}>z^{\beta
}z^{\gamma}e_{j}$ for all $\alpha,\beta,\gamma,i,j$, and $z^{\alpha}%
e_{i}>z^{\beta}e_{i}\Leftrightarrow z^{\alpha}e_{j}>z^{\beta}e_{j}$ for all
$\alpha,\beta,i,j$. Then $>$ induces a monomial ordering on the monomials of
$A$, which we again denote by $>$. In turn, any monomial ordering $>$ on $A$
induces two canonical monomial orderings on $A^{t}$, \emph{position over term}%
\begin{equation}
z^{\alpha}e_{i}>z^{\beta}e_{j}\ : \Longleftrightarrow \ i<j\text{ or (}i=j\text{
and }z^{\alpha}>z^{\beta}\text{)}%
\end{equation}
and analogously \emph{term over position}. We call a monomial ordering on
$A^{t}$ \emph{global} if the induced ordering on $A$ is global, that is, $z^{\alpha
}>1$ for all $\alpha$. Any $0\neq f\in A^{t}$ can be written as $f=c\cdot
z^{\alpha}e_{i}+g$ with $z^{\alpha}e_{i}>z^{\beta}e_{j}$ for all terms
$\widetilde{c}\cdot z^{\beta}e_{j}$ of $g\in A^{t}$. Then the term
$\operatorname{LT}_{>}(f)=c\cdot z^{\alpha}e_{i}$ is called the \emph{lead
term} of $f$, the constant $\operatorname{LC}_{>}(f)=c$ is called the
\emph{lead coefficient} of $f$, and $L_{>}(f)=z^{\alpha}e_{i}$ the \emph{lead
monomial} of $f$. The monomials of $A^{t}$ come with a natural partial order,
which we call \emph{divisibility}%
\begin{equation}
z^{\alpha}e_{i}\mid z^{\beta}e_{j}\ \Longleftrightarrow \ i=j\text{ and }%
z^{\alpha}\mid z^{\beta}.%
\end{equation}
For terms $c_{1}z^{\alpha}e_{i}$ and $c_{2}z^{\beta}e_{j}$ with $z^{\alpha
}e_{i}\mid z^{\beta}e_{j}$ we define their quotient as $\frac{c_{2}z^{\beta
}e_{j}}{c_{1}z^{\alpha}e_{i}}=\frac{c_{2}z^{\beta}}{c_{1}z^{\alpha}}\in A$.
Moreover,
we define the least common multiple $\operatorname{lcm}(z^{\alpha}e_{i},z^{\beta}e_{j})$ as zero if $i\neq j$, and as
$\operatorname{lcm}(z^{\alpha},z^{\beta})$ otherwise.
For a subset $G\subset A^{t}$, the \emph{leading module} $L(G)$ is the module
of all $A$-linear combinations of the lead monomials of the non-zero elements of
$G$.

By iteratively canceling the lead term of $f$ via multiples of lead terms of
the divisors in $G=\{g_{1},\ldots,g_{l}\}\subset A^{t}$ with respect to a fixed global ordering, we obtain a notion of
\emph{division with remainder} yielding a division expression
\begin{equation}
f=\sum_{i}a_{i}g_{i}+\operatorname*{NF}\nolimits_{>}(f,G)
\end{equation}
with $a_{i}\in A$ such that $\operatorname*{NF}(0,G)=0$, $\operatorname*{NF}%
_{>}(f,G)\neq0$ implies that $L_{>}(\operatorname*{NF}_{>}(f,G))\notin L(G)$,
and the lead monomial of $f$ is not smaller than that of any $a_{i}g_{i}$.

Let $U\subset A^{t}$ be a submodule and $>$ a global monomial ordering. A
finite set $0\notin G=\{g_{1},\ldots,g_{l}\}\subset U$ is called
\emph{Gr\"{o}bner basis }of $U$ with respect to $>$, if%
\begin{equation}
L_{>}(G)=L_{>}(U)\text{.}%
\end{equation}

\begin{theorem}
[Buchberger]With notation as above, the following conditions are equivalent:
\end{theorem}

\begin{enumerate}
\item $G$ is a Gr\"{o}bner basis of $U$,

\item $f\in U\Longleftrightarrow\operatorname*{NF}_{>}(f,G)=0$,

\item $U=\left\langle G\right\rangle $ and $\operatorname*{NF}_{>}%
(\operatorname*{spoly}_{>}(g_{i},g_{j}),G)=0$ for all $i\neq j$, where%
\begin{equation}
\operatorname*{spoly}(f,g):=\frac{\operatorname{lcm}(L(f),L(g))}{LT(f)}%
f-\frac{\operatorname{lcm}(L(f),L(g))}{LT(g)}g\text{.}%
\end{equation}
is the so-called \emph{S-polynomial} (or \emph{syzygy polynomial}) of $f$ and $g$.

\end{enumerate}

For a proof of this standard fact, see for example section 2.3 of ref.~\cite{GP}. A
generating set $G$ of $U$ can be extended to a Gr\"{o}bner basis by means of
Buchberger's algorithm, which according to the above criterion computes the remainder
$r=\operatorname*{NF}_{>}(\operatorname*{spoly}_{>}(g_{i},g_{j}),G)$ for all $g_i, g_j$ in $G$, adds
$r$ to $G$ if $r\neq0$, and iterates this process with the updated $G$ until all remainders vanish. This process terminates, since $A^t$ is Noetherian and,
hence, any ascending chain of submodules becomes stationary~(section
2.1 of ref.~\cite{GP}). Along this process, we can determine all relations between the
$g_{i}$:

\begin{algorithm}
[Syzygies]\label{alg syz} Let $M_{G}=(g_{1},\ldots,g_{l})\in A^{t\times l}$.
If $H$ is a Gr\"{o}bner basis of the column space of%
\begin{equation}
\left(
\begin{tabular}
[c]{c}%
$M_{G}$\\\hline
$%
\begin{array}
[c]{ccc}%
1\hspace{2mm} &  & 0\\
& \ddots & \\
0 &  & \hspace{2mm}1
\end{array}
$%
\end{tabular}
\right)
\end{equation}
with regard to the position over term order, $h_{1},\ldots,h_{m}$ are the
elements of $H$ in $\bigoplus_{i=t+1}^{t+l}e_{i}$, and $\pi:A^{t+l}\rightarrow
A^{l}$ is the projection onto the last $l$ coordinates, then
\begin{equation}
\operatorname*{syz}(g_{1},\ldots,g_{l}):=\ker M_{G}%
\end{equation}
is generated as an $A$-module by $\pi(h_{1}),\ldots,\pi(h_{m})$.
\end{algorithm}

So $\operatorname*{syz}(g_{1},\ldots,g_{l})=\operatorname*{im}\left(
\pi(h_{1}),\ldots,\pi(h_{m})\right)  $ is the image of the matrix with the
$\pi(h_{i})$ in the columns, in particular, $M_G \cdot \left(
\pi(h_{1}),\ldots,\pi(h_{m})\right)=0$. For a proof of correctness of the algorithm, see
for example lemma 2.5.3 of ref.~\cite {GP}. We can use this algorithm to compute
module intersections:

\begin{lemma}
[Intersection]\label{lem int}Let $M_{1}=\left\langle v_{1},\ldots
,v_{l}\right\rangle $ and $M_{2}=\left\langle w_{1},\ldots,w_{p}\right\rangle
$ be submodules of $A^{t}$, and
\begin{equation}
\operatorname*{syz}(v_{1},\ldots,v_{l},w_{1},\ldots,w_{p})=\operatorname*{im}%
\left(
\begin{array}
[c]{c}%
G\\
H
\end{array}
\right)
\end{equation}
with $G=(g_{i,j})\in A^{l\times m}$ and $H=(h_{i,j})\in A^{p\times m}$ as
obtained from Algorithm \ref{alg syz}. Then the columns of
\begin{equation}
(v_{1},\ldots,v_{l})\cdot G\text{,}%
\end{equation}
that is, the vectors $\sum_{i=1}^{l}g_{i,j}v_{i}$ with $j=1,\ldots,m$ generate
$M_{1}\cap M_{2}$.
\end{lemma}

\begin{proof}
Any element
\begin{equation}
s=\left(
\begin{array}
[c]{c}%
s_{1}\\
s_{2}%
\end{array}
\right)  \in\operatorname*{syz}(v_{1},\ldots,v_{l},w_{1},\ldots,w_{p})
\end{equation}
with $s_{1}=(a_{j})\in A^{l}$ and $s_{2}=(b_{j})\in A^{p}$ yields an element%
\begin{equation}
M_{1}\ni%
{\textstyle\sum\nolimits_{j=1}^{l}}
a_{j}v_{j}=-%
{\textstyle\sum\nolimits_{j=1}^{p}}
b_{j}w_{j}\in M_{2}%
\end{equation}
in $M_{1}\cap M_{2}$. On the other hand, if $m\in M_{1}\cap M_{2}$, then there
are $a_{j},b_{j}\in A$ with $m=%
{\textstyle\sum\nolimits_{j=1}^{l}}
a_{j}v_{j}=-%
{\textstyle\sum\nolimits_{j=1}^{p}}
b_{j}w_{j}$. Then the vertical concatenation of $s_{1}=(a_{j})\in A^{l}$ and
$s_{2}=(b_{j})\in A^{p}$ is in $\operatorname*{syz}(v_{1},\ldots,v_{l}%
,w_{1},\ldots,w_{p})$.
\end{proof}

Algorithm \ref{alg syz} in conjunction with Lemma \ref{lem int} can be used to
determine a generating system of $M_{1}\cap M_{2}$. For the module
intersection problems arising from the non-planar hexagon-box diagram however,
the performance of Buchberger's algorithm over $F$ is not sufficient to yield
a generating system in a reasonable time-frame. It turns out that a classical
technique (which to our knowledge dates back to ref.~\cite{GTZ}) to simulate computations
over rational function fields via polynomial computations is much faster. We
apply this technique to the Gr\"{o}bner basis computation yielding the syzygy
matrix used to determine the module intersection.

\begin{definition}
Given monomial orderings $>_{1}$ and $>_{2}$ on the monomials in $z_{1}%
,\ldots,z_{n}$ and $c_{1},\ldots,c_{r}$, respectively, a monomial ordering $>$
is given by%
\begin{equation}
z^{\alpha}c^{\beta}>z^{\alpha^{\prime}}c^{\beta^{\prime}}\
:\Longleftrightarrow \
z^{\alpha}>_{1}z^{\alpha^{\prime}}\text{ or (}z^{\alpha}=z^{\alpha^{\prime}%
}\text{ and }c^{\beta}>_{2}c^{\beta^{\prime}}\text{)}%
\end{equation}
We call $>$ the \emph{block ordering} $(>_{1},>_{2})$ associated to $>_{1}$and
$>_{2}$.
\end{definition}

\begin{lemma}
[Localization]\label{lem rat}Let $A=\mathbb{Q}(c_{1},\ldots,c_{r}%
)[z_{1},\ldots,z_{n}]$, let $B=\mathbb{Q}[z_{1},\ldots,z_{n},c_{1},\ldots,c_{r}]$, let $v_{1},\ldots,v_{l}$ be vectors with entries in
$B$, and define%
\begin{align*}
U  & =\left\langle v_{1},\ldots,v_{l}\right\rangle \subset A^{t}\\
U^{\prime}  & =\left\langle v_{1},\ldots,v_{l}\right\rangle \subset B^{t}%
\end{align*}
Let $G\subset B^{t}$ be a Gr\"{o}bner basis of $U^{\prime}$ with respect to a global
block ordering $(>_{1},>_{2})$ with blocks $z_{1},\ldots,z_{n}>c_{1}%
,\ldots,c_{r}$. Then $G$ is also a Gr\"{o}bner basis of $U$ with respect to
$>_{1}$.
\end{lemma}

\begin{proof} Denote the block ordering $(>_{1},>_{2})$ by $>$.
Every $f\in U$ can be written as%
\begin{equation}
f=\frac{1}{h}\sum_{i}\alpha_{i}v_{i}%
\end{equation}
with $\sum_{i}\alpha_{i}v_{i}\in U^{\prime}$ and $ h\in\mathbb{Q}[c_{1}%
,\ldots,c_{r}]$. By $h\cdot f\in U^{\prime}$ and $G$ being a Gr\"{o}bner basis
of $U^{\prime}$, we know that $\operatorname*{NF}(h\cdot f,G)=0$. Hence, there
is a $g\in G$ with $L_{>}(g)\mid L_{>}(h\cdot f)$ and $L_{>}(h\cdot
f)=L_{>}(h)\cdot L_{>}(f)$. Since $L_{>}(h)$ is a unit (invertible) in $A$ and
$>$ is a block ordering, we obtain that $L_{>_{1}}(g)\mid L_{>_{1}}(f)$.
This argument shows that $L_{>_{1}}(U)=L_{>_{1}}(G)$, that is, $G$ is a
Gr\"{o}bner basis of $U$ with respect to $>_{1}$.
\end{proof}


\begin{remark}
The method of Lemma \ref{lem rat} turns out to be efficient since the input
modules in our setting are homogeneous in the variables $z_{1},\ldots
,z_{n},c_{1},\ldots,c_{r}$, which allows for efficient sorting of the S-polynomials in Buchberger's algorithm by degree.
\end{remark}

\begin{remark}
For our setting, modular techniques, which compute over finite fields, combine
the results using the Chinese remainder theorem and apply rational
reconstruction (as developed in a general setting in refs.~\cite{BDFP} and
\cite{BDFLP}) seem not to be useful since very large constant coefficients occur. As a
result, this approach would require considering a large number of primes for lifting.
\end{remark}

The module intersection algorithm resulting from Lemma~\ref{lem int} and
Lemma~\ref{lem rat} produces generating sets which usually are not minimal in
any sense. For a homogeneous module, a minimal generating system can be
determined, however, the computation is very expensive. Another option is to
determine the unique reduced Gr\"{o}bner basis. A Gr\"{o}bner basis
$g_{1},\ldots,g_{l}$ is called \emph{minimal} if $L(g_{i})$ does not divide
$L(g_{j})$ for all $i\neq j$. Such a minimal Gr\"{o}bner basis is called
\emph{reduced} if none of the terms of the tails $g_{i}-\operatorname{LT}%
(g_{i})$ is divisible by some $L(g_{j})$. It does, however, also not make much
sense to pass to a minimal or reduced Gr\"{o}bner basis, since Gr\"{o}bner
bases can be much larger than generating systems and usually cannot be
obtained in a reasonable time in our setting. We hence employ a randomized
algorithm for trimming the generating systems to remove extraneous generators.
This algorithm is based on determining reduced Gr\"{o}bner bases after passing
to a finite field and specific values of parameters $c_{i}$:

\begin{algorithm}
[Trimming]\label{alg trim}Given a generating system $g_{1},\ldots,g_{l}$ of a
submodule $U\subset A^{t}$ with polynomial coefficients in $\mathbb{Z}%
[c_{1},\ldots,c_{r}]$, we proceed as follows:

\begin{enumerate}
\item Substitute the $c_{i}$ in the $g_{j}$ by pairwise different large prime
numbers $p_{i}$ obtaining polynomials $h_{j}\in\mathbb{Z}[z_{1},\ldots,z_{n}]$.

\item Choose a large prime $p$ different to the primes $p_i$. Apply the canonical map $\mathbb{Z}%
[z_{1},\ldots,z_{n}]\rightarrow\mathbb{F}_{p}[z_{1},\ldots,z_{n}]$ to the
$h_{j}$ obtaining $\overline{h_{1}},\ldots,\overline{h_{l}}$.

\item Choose an integer $j_{0}\in\{1,\ldots,l\}$.

\item Compute the reduced Gr\"{o}bner bases $G_{1}$ of
\begin{equation}
\left\langle \overline{h_{j}}\mid j=1,\ldots,l\right\rangle
\end{equation}
and $G_{2}$ of
\begin{equation}
\left\langle \overline{h_{j}}\mid j=1,\ldots,l\text{ with }j\neq
j_{0}\right\rangle .
\end{equation}

\item If $G_{1}=G_{2}$ return $\{g_{j}\mid j=1,\ldots,l$ and $j\neq
  j_{0}\}$ .
\end{enumerate}
\end{algorithm}

Multiple runs of the algorithm with different $p$ and $p_i$ reduce the chance of a bad prime or a bad parameter value.
We apply Algorithm \ref{alg trim} iteratively to drop generators, starting
with generators of large (byte) size.

\section{Sparse row reduction}
\label{section_REF}
In this section we turn to discussing linear algebra techniques.
Although the generation of IBP identities using eq.~\eqref{IBP_cut}
is very fast, achieving the reduction of target integrals to
linear combinations of master integrals is highly non-trivial.
This is because the latter step requires computing the row reduced echelon
form (RREF) of the IBP identities, which becomes computationally intensive
in cases with multiple external invariants (i.e., Mandelstam variables and mass parameters),
as these enter the IBP system as parameters.
Therefore, analytic computation of the RREF requires sophisticated linear algebra techniques.

\subsection{Selection of relevant and independent IBP identities}
The IBP identities generated from eq.~\eqref{IBP_cut} usually contain linearly
redundant identities, as well as identities which are irrelevant for reducing the
target integrals. To speed up linear reduction in the subsequent step,
we make use of the following methods to remove redundant and irrelevant identities.

\begin{enumerate}
\item \emph{Removal of redundant linear identities.} This can be done with the
  standard linear algebra algorithm of picking up independent rows of a
  matrix. We construct the matrix of all requisite IBP identities, sort the
  rows by their density (i.e., number of non-vanishing entries)
  or their byte count, and then compute the RREF of the transposed matrix numerically.
  The pivot locations correspond to the linearly independent IBP identities,
  giving preference to sparser IBP identities, or IBP identities of smaller sizes.
\item \emph{Removal of irrelevant linear identities.} Furthermore,
  given a target integral set, we can single out the relevant IBP identities
  which will ultimately reduce them to master integrals. This can also be done with
  a numeric RREF. We carry out the reduction numerically and record
  the rows used for reducing the targets in the computation
  (by recording the left-multiplying matrix of the row reduced matrix).
\end{enumerate}
Regarding the numeric RREF above, in practice we assign generic
integer values to all external invariants (Mandelstam variables and mass parameters) 
and the spacetime dimension and work with finite fields. The computation is powered by
the highly efficient sparse finite-field linear algebra package {\sc SpaSM}~\cite{spasm}.

\subsection{Sparse REF and RREF strategies}
We find that the IBP system that arises from eq. \eqref{IBP_cut}, after removing the
linearly dependent and irrelevant IBP identities for targets is, in general,
very sparse. To find the row echelon form (REF) and RREF efficiently,
it is crucial to apply a sophisticated pivoting strategy to keep the linear
system sparse in the intermediate steps.

First, we write the IBP identities in the form of a matrix, with the columns
sorted according to some integral ordering, for example like that of {\sc Azurite}.
Then the RREF will eventually reduce the target integrals
to the {\sc Azurite} master integral basis. However, it is important
to swap rows and columns during the REF computation, and find suitable
pivots for row reduction, in order to keep
the matrix sparse. This can be achieved by the heuristic Markowitz strategy
\cite{doi:10.1287/mnsc.3.3.255}, provided that all entries in the sparse matrix are of
a similar size. However, in our cases, the entries are polynomials
or rational functions in Mandelstam variables and
mass parameters, and so a weighted pivot strategy, considering
both the sparsity and the byte sizes, is used.

Note that we use a {\it total pivoting} strategy for which both row swaps
and column swaps are used. The row swaps will not change the final
result of the RREF, whereas the column swap {\it will} change the
final result of the RREF. This means that the target integrals will typically not be reduced to the
desired master integral basis, but a different integral basis. If we
require the target integrals to be reduced to a specific pre-determined basis
(say, the {\sc Azurite} basis), a basis change must be carried out after RREF computation. We find that it is
more efficient to allow both row and column swaps, and then compute
the basis change, than to allow row swap only (partial pivoting strategy).

The REF and RREF algorithm is implemented in our primitive
{\sc Mathematica} code. A more efficient implementation in {\sc Singular}
is in preparation and will become available soon.

\section{The non-planar hexagon-box diagram example}
\label{section_hexagon_box}
To demonstrate the power of our new method, we consider a cutting-edge
integration-by-parts reduction problem: the reduction of two-loop five-point
non-planar massless hexagon-box integrals. Recently, this diagram has
attracted a great deal of interest, and the hexagon-box integral with a chiral
numerator has been analytically computed by use of the bootstrap method and
by superconformal Ward identities \cite{Chicherin:2017dob,Chicherin:2018ubl}.
Here we consider the analytic IBP reduction of hexagon-box integrals with
arbitrary numerators with the degree up to four.

The hexagon-box diagram,
and the necessary cuts for deriving the IBP identities, are illustrated below in
figure~\ref{fig:spanning_set_of_cuts}.
\begin{figure}[ptb]
\begin{center}
\includegraphics[angle=0, width=0.85\textwidth]{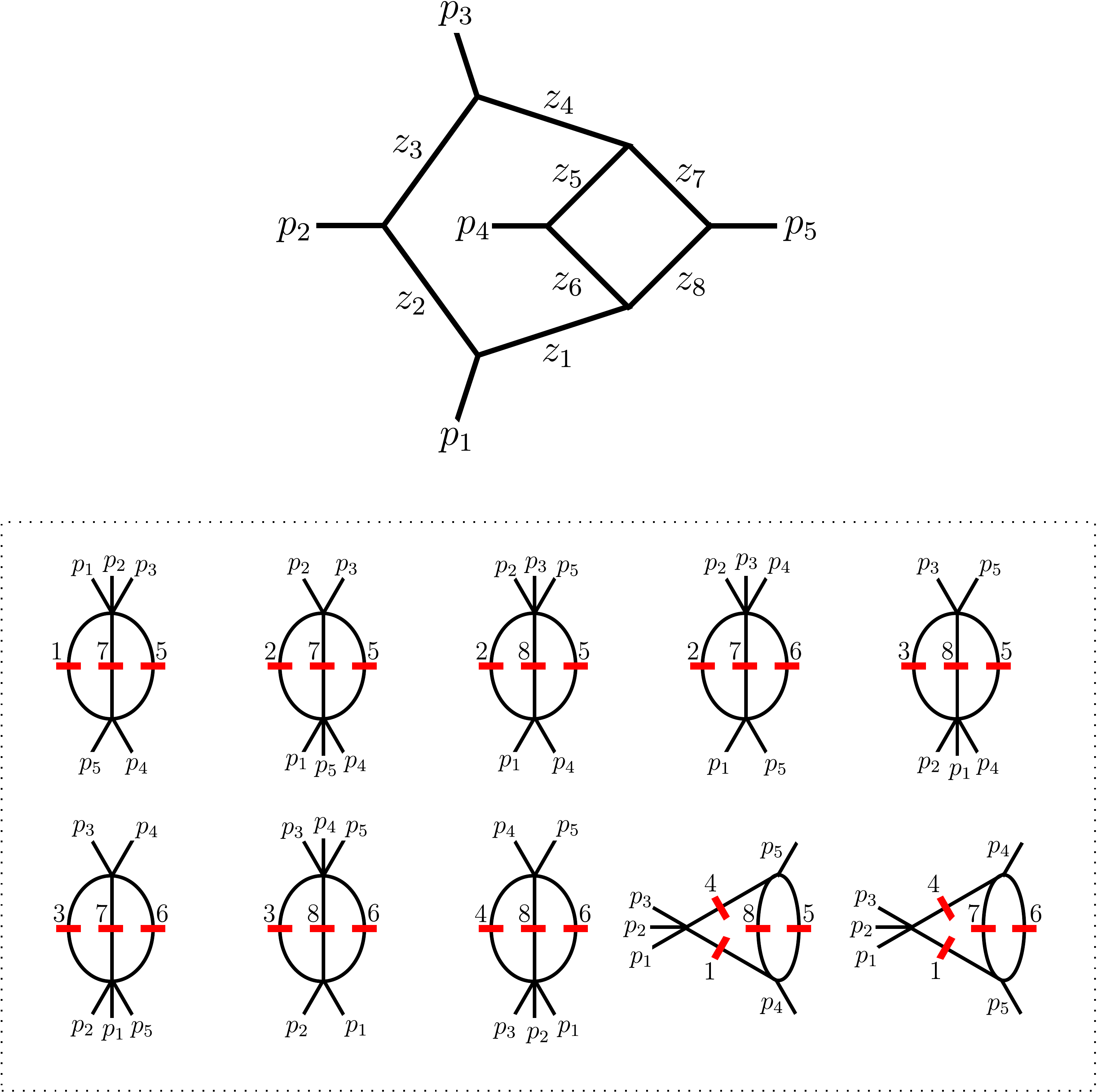}
{\vskip -2.5mm}
\caption{The fully massless non-planar hexagon-box diagram,
along with our labelling conventions for its internal lines.
The lower part shows the subset of the basis integrals
with the property that their graphs cannot be obtained by
adding internal lines to the graph of another basis integral.
The corresponding cuts $\{1,5,7\}$, $\{2,5,7\}$, $\{2,5,8\}$,
$\{2,6,7\}$, $\{3,5,8\}$, $\{3,6,7\}$, $\{3,6,8\}$, $\{4,6,8\}$,
$\{1,4,5,8\}$ and $\{1,4,6,7\}$ are the cuts required for deriving complete
IBP identities for the non-planar hexagon-box diagram.}
\label{fig:spanning_set_of_cuts}
\end{center}
\end{figure}

We define the inverse propagators as follows, setting
$P_{i\cdots j} \equiv p_i + \cdots + p_j$,
\begin{equation}
\begin{alignedat}{3}
   D_1&=\ell_1^2\,,               \hspace{8mm}  &&    D_2 = (\ell_1 - p_1)^2\,,       \hspace{8mm}  &&  D_3 = (\ell_1 - P_{12})^2\,, \\
   D_4&=(\ell_1-P_{123})^2\,,     \hspace{8mm}  &&    D_5 = (\ell_1+\ell_2+p_4)^2\,,  \hspace{8mm}  &&  D_6 = (\ell_1 + \ell_2)^2\,,  \\
   D_7&= (\ell_2 - p_5)^2\,,      \hspace{8mm}  &&    D_8 = \ell_2^2\,,               \hspace{8mm}  &&  D_9 = (\ell_1 + p_5)^2\,, \\
D_{10}&= (\ell_2 + p_1)^2 \,,                  \hspace{8mm}  && D_{11} =  (\ell_2 + p_2)^2    \,,              \hspace{8mm}  &&
\end{alignedat}
\end{equation}
and consider the family of Feynman integrals,
\begin{equation}
  I(\nu_1,\ldots, \nu_{11}) = \int \frac{\d^D l_1}{i \pi^{D/2}}
  \frac{\d^D l_2}{i \pi^{D/2}}  \frac{1}{D_1^{\nu_1} \cdots
    D_{11}^{\nu_{11}}}\,,
\label{Feynman_integral}
\end{equation}
with $\nu_9\leq 0$, $\nu_{10}\leq 0$ and $\nu_{11}\leq 0$. In terms of
the notation of subsection \ref{section:Module_intersection_no_cut}, we have
$L=2$, $E=4$, $k=8$ and $m=11$. We furthermore define the Baikov variables
$z_i\equiv D_i$, set $s_{ij}=(p_i+p_j)^2$ and express the IBP identities
in terms of the Mandelstam variables $s_{12}, s_{13}, s_{14}, s_{23}, s_{24}$.
\clearpage
Using {\sc Azurite} \cite{Georgoudis:2016wff} we establish that,
without applying global symmetries, there are $75$ ``pre''-master integrals---i.e.,
master integrals where global symmetry relations have not yet been imposed.
In the notation of eq.~\eqref{Feynman_integral}, they take the form,
\begin{eqnarray}
  &&\{{\mathcal I}_1,\ldots, {\mathcal I}_{75}\}\equiv\nonumber\\
  &&\begin{array}{lll}
   \big\{I(1, 1, 1, 1, 1, 1, 1, 1, -2, 0, 0), & I(1, 1, 1, 1, 1, 1, 1, 1, -1, 0, 0), &  I(1, 1, 1, 1,1, 1, 1, 1, 0, 0, 0), \\
 I(1, 1, 1, 0, 1, 1, 1, 1, -1, 0, 0), & I(1, 1, 1, -1, 1, 1, 1, 1, 0, 0, 0), & I(1, 1, 1, 0, 1, 1, 1, 1, 0, 0, 0), \\
I(1, 1, -1, 1, 1, 1, 1, 1, 0, 0, 0), & I(1, 1, 0, 1, 1, 1, 1, 1, 0, 0, 0), & I(1, -1, 1, 1, 1, 1, 1, 1, 0, 0, 0),\\
 I(1, 0, 1, 1, 1, 1, 1, 1, 0, 0, 0), & I(0, 1, 1, 1, 1, 1, 1, 1, -1, 0, 0), & I(-1, 1, 1, 1, 1, 1, 1, 1, 0, 0, 0), \\
I(0, 1, 1, 1, 1, 1, 1, 1, 0, 0, 0),& I(1, 1, 1, 1, 1, -1, 0, 1, 0, 0, 0), & I(1, 1, 1, 1, 1, 0, 0, 1, 0, 0, 0), \\
I(1, 1, 1, 1, -1, 1, 1, 0, 0, 0, 0),& I(1, 1, 1, 1, 0, 1, 1, 0, 0, 0, 0), & I(1, 1, 1, -1, 1, 1, 1, 0, 0, 0, 0),\\
 I(1, 1, 1, 0, 1, 1, 1, 0, 0, 0, 0), & I(1, 1, 1, -1, 1, 0, 1, 1, 0, 0, 0), & I(1, 1, 1, 0, 1, 0, 1, 1, 0, 0, 0), \\
I(1, -1, 1, 0, 1, 1, 1, 1, 0, 0, 0), & I(1, 0, 1, 0, 1, 1, 1, 1, 0, 0, 0), & I(1, 0, 0, 1, 1, 1, 1, 1, 0, 0, 0), \\
I(-1, 1, 1, 1, 1, 1, 0, 1, 0, 0, 0),& I(0, 1, 1, 1, 1, 1, 0, 1, 0, 0, 0), & I(-1, 1, 1, 1, 0, 1, 1, 1, 0, 0, 0),\\
 I(0, 1, 1, 1, 0, 1, 1, 1, 0, 0, 0),& I(0, 1, 1, 0, 1, 1, 1, 1, 0, 0, 0), & I(-1, 1, 0, 1, 1, 1, 1, 1, 0, 0, 0), \\
I(0, 1, 0, 1, 1, 1, 1, 1, 0, 0, 0), & I(1, 1, 1, 0, 1, 0, 1, 0, 0, 0, 0), & I(1, 1, 1, 0, 1, 0, 0, 1, 0, 0, 0),\\
 I(1, 1, 1, 0, 0, 1, 1, 0, 0, 0, 0), & I(1, 1, 0, 1, 1, 0, 0, 1, 0, 0, 0), & I(1, 1, 0, 1, 0, 1, 1, 0, 0, 0, 0), \\
I(1, 1, 0, 0, 1, 1, 1, 0, 0, 0, 0), & I(1, 1, 0, 0, 1, 0, 1, 1, 0, 0, 0), & I(1, 0, 1, 1, 1, 0, 0, 1, 0, 0, 0), \\
I(1, 0, 1, 1, 0, 1, 1, 0, 0, 0, 0), & I(1, -1, 1, 0, 1, 1, 1, 0, 0, 0, 0), & I(1, 0, 1, 0, 1, 1, 1, 0, 0, 0, 0),\\
 I(1, -1, 1, 0, 1, 0, 1, 1, 0, 0, 0), & I(1, 0, 1, 0, 1, 0, 1, 1, 0, 0, 0), & I(0, 1, 1, 1, 1, 0, 0, 1, 0, 0, 0), \\
I(0, 1, 1, 1, 0, 1, 1, 0, 0, 0, 0), & I(0, 1, 1, 1, 0, 1, 0, 1, 0, 0, 0), & I(0, 1, 1, 0, 1, 1, 1, 0, 0, 0, 0),\\
 I(0, 1, 1, 0, 1, 1, 0, 1, 0, 0, 0), & I(0, 1, 1, 0, 1, 0, 1, 1, 0, 0, 0), & I(0, 1, 1, 0, 0, 1, 1, 1, 0, 0, 0),\\
 I(-1, 1, 0, 1, 1, 1, 0, 1, 0, 0, 0),& I(0, 1, 0, 1, 1, 1, 0, 1, 0, 0, 0), & I(-1, 1, 0, 1, 0, 1, 1, 1, 0, 0, 0),\\
 I(0, 1, 0, 1, 0, 1, 1, 1, 0, 0, 0), & I(0, 1, 0, 0, 1, 1, 1, 1, 0, 0, 0), & I(0, 0, 1, 1, 1, 1, 0, 1, 0, 0, 0), \\
I(0, 0, 1, 1, 0, 1, 1, 1, 0, 0, 0), & I(0, 0, 1, 0, 1, 1, 1, 1, 0, 0, 0), & I(1, 0, 1, 0, 1, 0, 1, 0, 0, 0, 0), \\
I(1, 0, 1, 0, 1, 0, 0, 1, 0, 0, 0), & I(1, 0, 1, 0, 0, 1, 1, 0, 0, 0, 0), & I(1, 0, 0, 1, 1, 0, 0, 1, 0, 0, 0), \\
I(1, 0, 0, 1, 0, 1, 1, 0, 0, 0, 0), & I(0, 1, 0, 1, 1, 0, 0, 1, 0, 0, 0), & I(0, 1, 0, 1, 0, 1, 1, 0, 0, 0, 0),\\
 I(0, 1, 0, 1, 0, 1, 0, 1, 0, 0, 0), & I(1, 0, 0, 0, 1, 0, 1, 0, 0, 0, 0), & I(0, 1, 0, 0, 1, 0, 1, 0, 0, 0, 0),\\
 I(0, 1, 0, 0, 1, 0, 0, 1, 0, 0, 0), & I(0, 1, 0, 0, 0, 1, 1, 0, 0, 0, 0), & I(0, 0, 1, 0, 1, 0, 0, 1, 0, 0, 0),\\
 I(0, 0, 1, 0, 0, 1, 1, 0, 0, 0, 0), & I(0, 0, 1, 0, 0, 1, 0, 1, 0, 0, 0), & I(0, 0, 0, 1, 0, 1, 0, 1, 0, 0, 0) \big\} \,.
  \end{array} \nonumber
\\ \label{eq:Azurite_basis}
\end{eqnarray}
The graphs of these integrals are shown in figure~\ref{fig:integral_basis}.
We remark that {\sc Azurite} chooses master integrals which contain no doubled propagators.

\clearpage

\begin{figure}[!h]
\begin{center}
\includegraphics[angle=0, width=\textwidth]{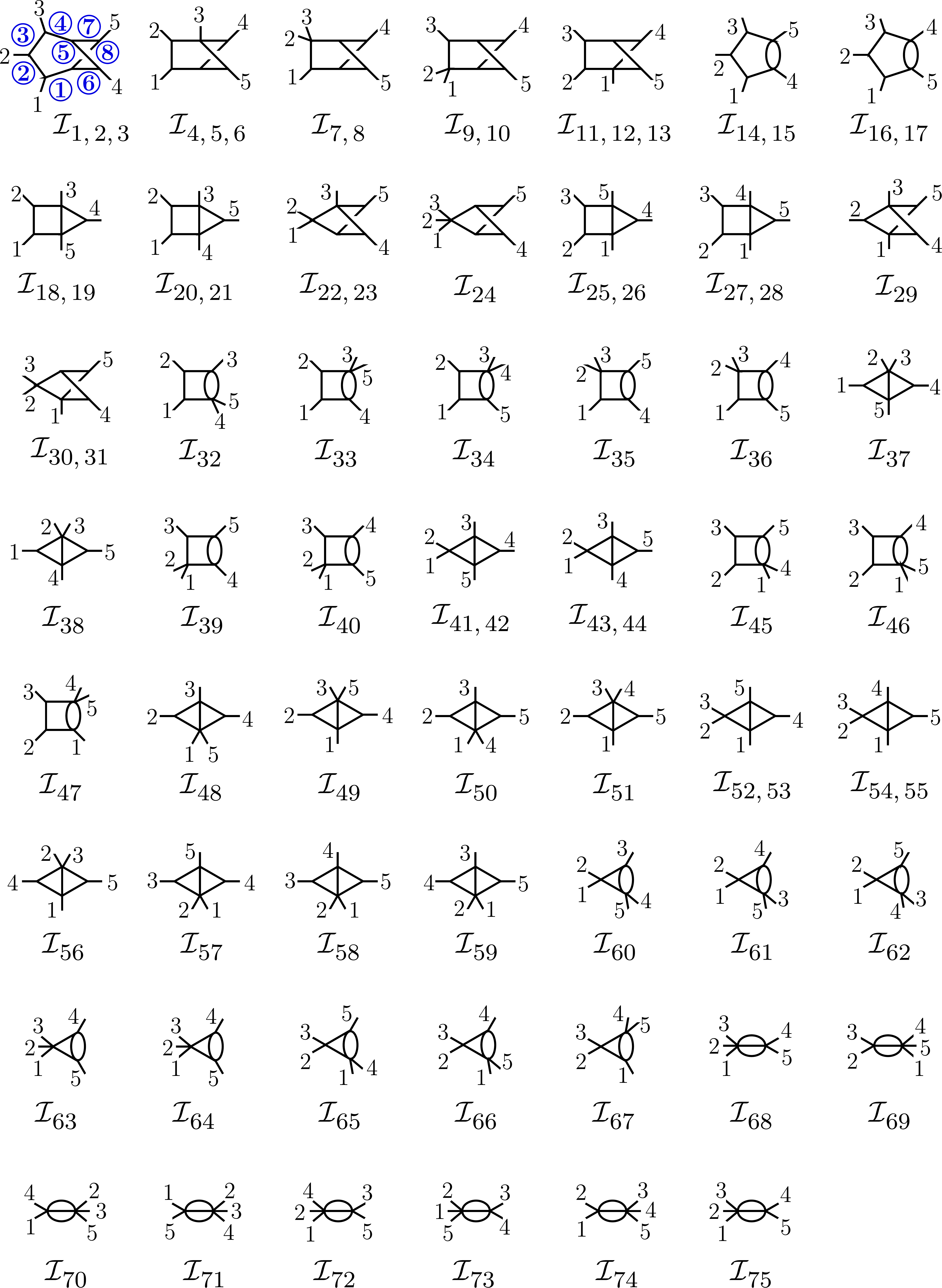}
{\vskip -2.5mm}
\caption{(Color online.) The $75$ ``pre''-master integrals found by {\sc Azurite},
  with the global symmetry option turned off, for
the non-planar hexagon-box family of eq.~\eqref{Feynman_integral}. By
turning on the global symmetry option, {\sc Azurite} determines that
there are $73$ master integrals. These are the illustrated ``pre''-master integrals
excluding $\mathcal{I}_{63}$ and $\mathcal{I}_{68}$. Our labeling convention for the
propagators, corresponding to the indices recorded in eq.~\eqref{eq:Azurite_basis},
are shown in the graph of the hexagon-box itself as the blue encircled numbers.}
\label{fig:integral_basis}
\end{center}
\end{figure}

\clearpage

As explained at the end of section \ref{subsection_module_intersection_on_cuts},
the cuts that are necessary to construct the complete IBP identities are the maximal cuts
of the ``uncollapsible'' master integrals in eq.~\eqref{eq:Azurite_basis}. From
this list of master integrals we find that the following $10$ cuts
are necessary for computing the complete IBP identities,
\begin{gather}
  \{1,5,7\}, \{2,5,7\}, \{2,5,8\}, \{2,6,7\}, \{3,5,8\}, \{3,6,7\}, \{3,6,8\}, \{4,6,8\}, \nonumber \\
\{1,4,5,8\} , \{1,4,6,7\}  \label{eq:spanning_set_of_cuts} \,.
\end{gather}
Here for example the notation $\{1,5,7\}$ means $\mathcal{S}_\mathrm{cut} = (\zeta_1,\zeta_2,\zeta_3)=(1,5,7)$; that is,
the triple cut,
\begin{equation}
  \label{eq:9}
  z_1 \to 0,\quad z_5 \to 0, \quad z_7 \to 0 \,.
\end{equation}
Thus, for the hexagon-box diagram, the set of necessary cuts thus includes $8$ triple cuts
and $2$ quadruple cuts. The ``pre''-master integrals
supported on these cuts are listed in table~\ref{table:pre-master}.

The advantage of applying cuts is that the number of integrals in the IBP relations
on each cut will be much less than when no cuts are applied.

\begin{table}[ht]
  \centering
  \begin{tabular}{|c|p{12cm}|}
\hline
cut & ``pre''-master integrals\\
\hline
   $\{1,5,7\}$ & $\mathcal I_{1}$, $\mathcal I_{2}$, $\mathcal I_{3}$,
                 $\mathcal I_{4}$, $\mathcal I_{5}$, $\mathcal I_{6}$,
                 $\mathcal I_{7}$, $\mathcal I_{8}$, $\mathcal I_{9}$,
                 $\mathcal I_{10}$, $\mathcal I_{18}$, $\mathcal
                 I_{19}$, $\mathcal I_{20}$, $\mathcal I_{21}$,
                 $\mathcal I_{22}$, $\mathcal I_{23}$, $\mathcal
                 I_{24}$, $\mathcal I_{32}$, $\mathcal I_{37}$,
                 $\mathcal I_{38}$, $\mathcal I_{41}$, $\mathcal
                 I_{42}$, $\mathcal I_{43}$, $\mathcal I_{44}$,
                 $\mathcal I_{60}$, $\mathcal I_{68}$ \\
\hline
$\{2,5,7\}$ & $\mathcal I_{1}$, $\mathcal I_{2}$, $\mathcal I_{3}$,
         $\mathcal I_{4}$, $\mathcal I_{5}$, $\mathcal I_{6}$,
         $\mathcal I_{7}$, $\mathcal I_{8}$, $\mathcal I_{11}$,
         $\mathcal I_{12}$, $\mathcal I_{13}$, $\mathcal I_{18}$,
         $\mathcal I_{19}$, $\mathcal I_{20}$, $\mathcal I_{21}$,
         $\mathcal I_{29}$, $\mathcal I_{30}$, $\mathcal I_{31}$,
         $\mathcal I_{32}$, $\mathcal I_{37}$, $\mathcal I_{38}$,
         $\mathcal I_{48}$, $\mathcal I_{50}$, $\mathcal I_{56}$,
         $\mathcal I_{69}$ \\
\hline
$\{2, 5, 8\}$ & $\mathcal I_{1}$, $\mathcal I_{2}$, $\mathcal
                  I_{3}$, $\mathcal I_{4}$, $\mathcal I_{5}$,
                  $\mathcal I_{6}$, $\mathcal I_{7}$, $\mathcal
                  I_{8}$, $\mathcal I_{11}$, $\mathcal I_{12}$,
                  $\mathcal I_{13}$, $\mathcal I_{14}$, $\mathcal
                  I_{15}$, $\mathcal I_{20}$, $\mathcal I_{21}$,
                  $\mathcal I_{25}$, $\mathcal I_{26}$, $\mathcal
                  I_{29}$, $\mathcal I_{30}$, $\mathcal I_{31}$,
                  $\mathcal I_{33}$, $\mathcal I_{35}$, $\mathcal
                  I_{38}$, $\mathcal I_{45}$, $\mathcal I_{49}$,
                  $\mathcal I_{50}$, $\mathcal I_{52}$, $\mathcal
                  I_{53}$, $\mathcal I_{56}$, $\mathcal I_{65}$,
                  $\mathcal I_{70}$ \\
\hline
$\{2, 6, 7\}$ & $\mathcal I_{1}$, $\mathcal I_{2}$, $\mathcal I_{3}$,
                 $\mathcal I_{4}$, $\mathcal I_{5}$, $\mathcal I_{6}$,
                 $\mathcal I_{7}$, $\mathcal I_{8}$, $\mathcal
                 I_{11}$, $\mathcal I_{12}$, $\mathcal I_{13}$,
                 $\mathcal I_{16}$, $\mathcal I_{17}$, $\mathcal
                 I_{18}$, $\mathcal I_{19}$, $\mathcal I_{27}$,
                 $\mathcal I_{28}$, $\mathcal I_{29}$, $\mathcal
                 I_{30}$, $\mathcal I_{31}$, $\mathcal I_{34}$,
                 $\mathcal I_{36}$, $\mathcal I_{37}$, $\mathcal
                 I_{46}$, $\mathcal I_{48}$, $\mathcal I_{51}$,
                 $\mathcal I_{54}$, $\mathcal I_{55}$, $\mathcal
                 I_{56}$, $\mathcal I_{66}$, $\mathcal I_{71}$ \\
\hline
$\{3,5,8\}$ & $\mathcal I_{1}$, $\mathcal I_{2}$, $\mathcal I_{3}$,
              $\mathcal I_{4}$, $\mathcal I_{5}$, $\mathcal I_{6}$,
              $\mathcal I_{9}$, $\mathcal I_{10}$, $\mathcal I_{11}$,
              $\mathcal I_{12}$, $\mathcal I_{13}$, $\mathcal I_{14}$,
              $\mathcal I_{15}$, $\mathcal I_{20}$, $\mathcal I_{21}$,
              $\mathcal I_{22}$, $\mathcal I_{23}$, $\mathcal I_{25}$,
              $\mathcal I_{26}$, $\mathcal I_{29}$, $\mathcal I_{33}$,
              $\mathcal I_{39}$, $\mathcal I_{43}$, $\mathcal I_{44}$,
              $\mathcal I_{45}$, $\mathcal I_{49}$, $\mathcal I_{50}$,
              $\mathcal I_{57}$, $\mathcal I_{59}$, $\mathcal I_{61}$,
              $\mathcal I_{72}$ \\
\hline
$\{3,6,7\}$ & $\mathcal I_{1}$, $\mathcal I_{2}$, $\mathcal I_{3}$,
              $\mathcal I_{4}$, $\mathcal I_{5}$, $\mathcal I_{6}$,
              $\mathcal I_{9}$, $\mathcal I_{10}$, $\mathcal I_{11}$,
              $\mathcal I_{12}$, $\mathcal I_{13}$, $\mathcal I_{16}$,
              $\mathcal I_{17}$, $\mathcal I_{18}$, $\mathcal I_{19}$,
              $\mathcal I_{22}$, $\mathcal I_{23}$, $\mathcal I_{27}$,
              $\mathcal I_{28}$, $\mathcal I_{29}$, $\mathcal I_{34}$,
              $\mathcal I_{40}$, $\mathcal I_{41}$, $\mathcal I_{42}$,
              $\mathcal I_{46}$, $\mathcal I_{48}$, $\mathcal I_{51}$,
              $\mathcal I_{58}$, $\mathcal I_{59}$, $\mathcal I_{62}$,
              $\mathcal I_{73}$\\
\hline
$\{3,6,8\}$ & $\mathcal I_{1}$, $\mathcal I_{2}$, $\mathcal I_{3}$,
              $\mathcal I_{4}$, $\mathcal I_{5}$, $\mathcal I_{6}$,
              $\mathcal I_{9}$, $\mathcal I_{10}$, $\mathcal I_{11}$,
              $\mathcal I_{12}$, $\mathcal I_{13}$, $\mathcal I_{22}$,
              $\mathcal I_{23}$, $\mathcal I_{25}$, $\mathcal I_{26}$,
              $\mathcal I_{27}$, $\mathcal I_{28}$, $\mathcal I_{29}$,
              $\mathcal I_{47}$, $\mathcal I_{49}$, $\mathcal I_{51}$,
              $\mathcal I_{57}$, $\mathcal I_{58}$, $\mathcal I_{59}$,
              $\mathcal I_{74}$ \\
\hline
$\{4,6,8\}$ & $\mathcal I_{1}$, $\mathcal I_{2}$, $\mathcal I_{3}$,
              $\mathcal I_{7}$, $\mathcal I_{8}$, $\mathcal I_{9}$,
              $\mathcal I_{10}$, $\mathcal I_{11}$, $\mathcal I_{12}$,
              $\mathcal I_{13}$, $\mathcal I_{24}$, $\mathcal I_{25}$,
              $\mathcal I_{26}$, $\mathcal I_{27}$, $\mathcal I_{28}$,
              $\mathcal I_{30}$, $\mathcal I_{31}$, $\mathcal I_{47}$,
              $\mathcal I_{52}$, $\mathcal I_{53}$, $\mathcal I_{54}$,
              $\mathcal I_{55}$, $\mathcal I_{57}$, $\mathcal I_{58}$,
              $\mathcal I_{67}$, $\mathcal I_{75}$\\
\hline
$\{1,4,5,8\}$ & $\mathcal I_{1}$, $\mathcal I_{2}$, $\mathcal I_{3}$,
                $\mathcal I_{7}$, $\mathcal I_{8}$, $\mathcal I_{9}$,
                $\mathcal I_{10}$, $\mathcal I_{14}$, $\mathcal
                I_{15}$, $\mathcal I_{24}$, $\mathcal I_{35}$,
                $\mathcal I_{39}$, $\mathcal I_{63}$ \\
\hline
$\{1,4,6,7\}$ & $\mathcal I_{1}$, $\mathcal I_{2}$, $\mathcal I_{3}$,
                $\mathcal I_{7}$, $\mathcal I_{8}$, $\mathcal I_{9}$,
                $\mathcal I_{10}$, $\mathcal I_{16}$, $\mathcal
                I_{17}$, $\mathcal I_{24}$, $\mathcal I_{36}$,
                $\mathcal I_{40}$, $\mathcal I_{64}$\\
\hline
\end{tabular}
  \caption{``Pre''-master integrals supported on each of the $10$ cuts
  necessary to construct the complete IBP reductions.}
\label{table:pre-master}
\end{table}

If we switch on global symmetries in {\sc Azurite}, then the number
of independent master integrals is found to be $73$. This is due to the additional
symmetry relations,
\begin{align}
\mathcal I_{63}\equiv I(1, 0, 0, 1, 1, 0, 0, 1, 0, 0, 0) &=I(1, 0, 0, 1, 0, 1, 1, 0, 0, 0,
                                     0) \,,\nonumber \\
\mathcal I_{68}\equiv I(1, 0, 0, 0, 1, 0, 1, 0, 0, 0, 0) &=   I(0, 0, 0, 1, 0, 1, 0, 1, 0, 0, 0) \,.
\label{global_symmetry}
\end{align}
Hence the set of $73$ master integrals is,
\begin{equation}
  \label{eq:10}
   \mathcal{F}_\mathrm{MI} = \{ \mathcal{I}_1, \ldots, \mathcal{I}_{75} \} \setminus \{ \mathcal I_{63}, \mathcal I_{68}\} \,.
\end{equation}

In this paper, we first reduce integrals to the linear combination of
$75$ ``pre''-master integrals and then apply
eq.~\eqref{global_symmetry} to further achieve the reduction to the
linearly independent $73$ master integrals.

To demonstrate the power of our method, we show how to reduce all the
numerator-degree-$4$, $3$, $2$ and $1$ hexagon-box integrals (our target integrals),
\begin{eqnarray}
  \begin{array}{lll}
    \{I(1,1,1,1,1,1,1,1,0,0,-4),&I(1,1,1,1,1,1,1,1,0,-1,-3),&I(1,1,1,1,1,1,1,1,0,-2,-2),\\
I(1,1,1,1,1,1,1,1,0,-3,-1) ,&I(1,1,1,1,1,1,1,1,0,-4,0),&I(1,1,1,1,1,1,1,1,-1,0,-3),\\
I(1,1,1,1,1,1,1,1,-1,-1,-2), & I(1,1,1,1,1,1,1,1,-1,-2,-1),&I(1,1,1,1,1,1,1,1,-1,-3,0),\\
I(1,1,1,1,1,1,1,1,-2,0,-2),&I(1,1,1,1,1,1,1,1,-2,-1,-1),&I(1,1,1,1,1,1,1,1,-2,-2,0),\\
I(1,1,1,1,1,1,1,1,-3,0,-1),&I(1,1,1,1,1,1,1,1,-3,-1,0),&I(1,1,1,1,1,1,1,1,-4,0,0),\\
I(1,1,1,1,1,1,1,1,0,0,-3),&I(1,1,1,1,1,1,1,1,0,-1,-2),&I(1,1,1,1,1,1,1,1,0,-2,-1),\\
I(1,1,1,1,1,1,1,1,0,-3,0),&I(1,1,1,1,1,1,1,1,-1,0,-2),&I(1,1,1,1,1,1,1,1,-1,-1,-1),\\
I(1,1,1,1,1,1,1,1,-1,-2,0),&I(1,1,1,1,1,1,1,1,-2,0,-1),&I(1,1,1,1,1,1,1,1,-2,-1,0),\\
I(1,1,1,1,1,1,1,1,-3,0,0),&I(1,1,1,1,1,1,1,1,0,0,-2),&I(1,1,1,1,1,1,1,1,0,-1,-1),\\
I(1,1,1,1,1,1,1,1,0,-2,0),&I(1,1,1,1,1,1,1,1,-1,0,-1),&I(1,1,1,1,1,1,1,1,-1,-1,0),\\
I(1,1,1,1,1,1,1,1,0,0,-1),&I(1,1,1,1,1,1,1,1,0,-1,0)\}
  \end{array}\nonumber\\
\label{eq:target}
\end{eqnarray}
analytically to express them as linear combinations of the $73$ master integrals.

\subsection{Module intersection on cuts}
In this section we show explicitly how to apply our module intersection method on
cuts, in order to obtain simplified IBP systems (i.e., which do not involve
integrals with doubled propagators) on unitarity cuts.

The module $M_1$ defined in subsection~\ref{section:Module_intersection_no_cut}
for the hexagon-box diagram without cuts applied, is
generated by the following $2\times (4+2)=12$ generators, cf.~eq.~\eqref{M1_generators},

{\tiny
\begin{gather}
(z_1-z_2,z_1-z_2,-s_{12}+z_1-z_2,-s_{12}-s_{13}+z_1-z_2,s_{14}+z_1-z_2-z_8+z_{10},z_1-z_2-z_8+z_{10},0,0,-s_{12}-s_{13}-s_{14}+z_1-z_2,0,0)\nonumber
\\
(0,0,0,0,s_{14}+z_1-z_2-z_8+z_{10},z_1-z_2-z_8+z_{10},s_{12}+s_{13}+s_{14}-z_8+z_{10},z_{10}-z_8,0,z_{10}-z_8,s_{12}-z_8+z_{10})\nonumber
\\
(s_{12}+z_2-z_3,z_2-z_3,z_2-z_3,-s_{23}+z_2-z_3,s_{12}+s_{24}+z_2-z_3-z_8+z_{11},s_{12}+z_2-z_3-z_8+z_{11},0,0,-s_{23}-s_{24}+z_2-z_3,0,0)\nonumber
\\
(0,0,0,0,s_{12}+s_{24}+z_2-z_3-z_8+z_{11},s_{12}+z_2-z_3-z_8+z_{11},s_{12}+s_{23}+s_{24}-z_8+z_{11},z_{11}-z_8,0,s_{12}-z_8+z_{11},z_{11}-z_8)\nonumber
\\  (s_{13}+s_{23}+z_3-z_4,s_{23}+z_3-z_4,z_3-z_4,z_3-z_4,-2
s_{12}-s_{13}-s_{14}-s_{23}-s_{24}+z_3-z_5+z_6+z_7+z_8-z_9-z_{10}-z_{11},\nonumber
\\ -s_{12}+z_3-z_5+z_6+z_7+z_8-z_9-z_{10}-z_{11},0,0,s_{12}+s_{13}+s_{14}+s_{23}+s_{24}+z_3-z_4,0,0)\nonumber
\\  (0,0,0,0,-2
s_{12}-s_{13}-s_{14}-s_{23}-s_{24}+z_3-z_5+z_6+z_7+z_8-z_9-z_{10}-z_{11},-s_{12}+z_3-z_5+z_6+z_7+z_8-z_9-z_{10}-z_{11},\nonumber
\\
-2
s_{12}-s_{13}-s_{14}-s_{23}-s_{24}+z_4-z_5+z_6+z_7+z_8-z_9-z_{10}-z_{11},-s_{12}-s_{13}-s_{23}+z_4-z_5+z_6+z_7+z_8-z_9-z_{10}-z_{11},
\nonumber
\\
0,-s_{12}-s_{23}+z_4-z_5+z_6+z_7+z_8-z_9-z_{10}-z_{11},-s_{12}-s_{13}+z_4-z_5+z_6+z_7+z_8-z_9-z_{10}-z_{11})\nonumber
\\
(-s_{12}-s_{13}-s_{23}+z_4-z_9,-s_{12}-s_{13}-s_{14}-s_{23}+z_4-z_9,-s_{12}-s_{13}-s_{14}-s_{23}-s_{24}+z_4-z_9,z_4-z_9,z_5-z_6,z_5-z_6,0,0,z_4-z_9,0,0)\nonumber
\\
(0,0,0,0,z_5-z_6,z_5-z_6,-z_4+z_5-z_6+z_9,s_{12}+s_{13}+s_{23}-z_4+z_5-z_6+z_9, \nonumber
\\ 0,s_{12}+s_{13}+s_{14}+s_{23}-z_4+z_5-z_6+z_9,s_{12}+s_{13}+s_{23}+s_{24}-z_4+z_5-z_6+z_9)\nonumber
\\  (2
z_1,z_1+z_2,-s_{12}+z_1+z_3,-s_{12}-s_{13}-s_{23}+z_1+z_4,-s_{12}-s_{13}-s_{23}+z_1+z_4+z_6-z_8-z_9,z_1+z_6-z_8,0,0,z_1+z_9,0,0)\nonumber
\\
(0,0,0,0,-s_{12}-s_{13}-s_{23}+z_1+z_4+z_6-z_8-z_9,z_1+z_6-z_8,z_6-z_8-z_9,-z_1+z_6-z_8,0,-z_2+z_6-z_8,s_{12}-z_1+z_2-z_3+z_6-z_8)\nonumber
\\
(-z_1+z_6-z_8,-z_1+z_6-z_{10},-z_1+z_6+z_8-z_{10}-z_{11},s_{12}+s_{13}+s_{23}-z_1-z_4+z_5-z_7+z_9,\nonumber
\\ s_{12}+s_{13}+s_{23}-z_1-z_4+z_5+z_8+z_9,-z_1+z_6+z_8,0,0,-z_1+z_6-z_7,0,0)\nonumber
\\
(0,0,0,0,s_{12}+s_{13}+s_{23}-z_1-z_4+z_5+z_8+z_9,-z_1+z_6+z_8,z_7+z_8,2
z_8,0,z_8+z_{10},z_8+z_{11})\,.
\end{gather}
}
Note that the generators are at most linear in the $z_i$, and always homogeneous in
$z_1, \ldots, z_{11}$ and $s_{12}, s_{13}, s_{14}, s_{23}, s_{24}$.

The module $M_2$ for the hexagon-box diagram without applied cuts is
generated by the following $11$ generators, cf.~eq.~\eqref{M2_generators},
\begin{align}
  \label{eq:11}
 &(z_1,0,0,0,0,0,0,0,0,0,0), (0,z_2,0,0,0,0,0,0,0,0,0)\nonumber \\ & (0,0,z_3,0,0,0,0,0,0,0,0),(0,0,0,z_4,0,0,0,0,0,0,0)\nonumber \\ & (0,0,0,0,z_5,0,0,0,0,0,0), (0,0,0,0,0,z_6,0,0,0,0,0)\nonumber \\ & (0,0,0,0,0,0,z_7,0,0,0,0), (0,0,0,0,0,0,0,z_8,0,0,0)\nonumber \\&  (0,0,0,0,0,0,0,0,1,0,0), (0,0,0,0,0,0,0,0,0,1,0)\nonumber \\ & (0,0,0,0,0,0,0,0,0,0,1)\,.
\end{align}
We now proceed to consider the modules on the unitarity cuts given in eq.~\eqref{eq:spanning_set_of_cuts}.
For example, for the cut $\{2,5,7\}$, we apply the replacements
\begin{equation}
z_2 \to 0\,,\quad z_5 \to 0\,, \quad z_7 \to 0\,.
\label{cut_257}
 \end{equation}
The propagator indices are thus classified as,
\begin{align}
\mathcal{S}_\mathrm{cut}     =  \{ 2,5,7\}\,, \quad
\mathcal{S}_\mathrm{uncut}   =  \{ 1,3,4,6,8\} \,, \quad
\mathcal{S}_\mathrm{ISP}     =  \{ 9,10,11\}  \,.
\end{align}
From eq.~\eqref{module_cut} it follows that the generators of $\tilde M_1$ on the cut
$\{2,5,7\}$ take the form,
{\tiny
\begin{gather}
  \label{eq:8}
   (z_1,z_1,z_1-s_{12},-s_{12}-s_{13}+z_1,s_{14}+z_1-z_8+z_{10},z_1-z_8+z_{10},0,0,-s_{12}-s_{13}-s_{14}+z_1,0,0)\nonumber
   \\
(0,0,0,0,s_{14}+z_1-z_8+z_{10},z_1-z_8+z_{10},s_{12}+s_{13}+s_{14}-z_8+z_{10},z_{10}-z_8,0,z_{10}-z_8,s_{12}-z_8+z_{10})\nonumber
\\
(s_{12}-z_3,-z_3,-z_3,-s_{23}-z_3,s_{12}+s_{24}-z_3-z_8+z_{11},s_{12}-z_3-z_8+z_{11},0,0,-s_{23}-s_{24}-z_3,0,0)\nonumber
\\
 (0,0,0,0,s_{12}+s_{24}-z_3-z_8+z_{11},s_{12}-z_3-z_8+z_{11},s_{12}+s_{23}+s_{24}-z_8+z_{11},z_{11}-z_8,0,s_{12}-z_8+z_{11},z_{11}-z_8)\nonumber
 \\
(s_{13}+s_{23}+z_3-z_4,s_{23}+z_3-z_4,z_3-z_4,z_3-z_4,-2
s_{12}-s_{13}-s_{14}-s_{23}-s_{24}+z_3+z_6+z_8-z_9-z_{10}-z_{11}, %
 \nonumber
\\ -s_{12}+z_3+z_6+z_8-z_9-z_{10}-z_{11},0,0,s_{12}+s_{13}+s_{14}+s_{23}+s_{24}+z_3-z_4,0,0)\nonumber
\\
 (0,0,0,0,-2
 s_{12}-s_{13}-s_{14}-s_{23}-s_{24}+z_3+z_6+z_8-z_9-z_{10}-z_{11},-s_{12}+z_3+z_6+z_8-z_9-z_{10}-z_{11},%
\nonumber\\
-2
 s_{12}-s_{13}-s_{14}-s_{23}-s_{24}+z_4+z_6+z_8-z_9-z_{10}-z_{11},-s_{12}-s_{13}-s_{23}+z_4+z_6+z_8-z_9-z_{10}-z_{11},%
\nonumber\\
0,-s_{12}-s_{23}+z_4+z_6+z_8-z_9-z_{10}-z_{11},-s_{12}-s_{13}+z_4+z_6+z_8-z_9-z_{10}-z_{11})\nonumber
 \\
 (-s_{12}-s_{13}-s_{23}+z_4-z_9,-s_{12}-s_{13}-s_{14}-s_{23}+z_4-z_9,-s_{12}-s_{13}-s_{14}-s_{23}-s_{24}+z_4-z_9,z_4-z_9,-z_6,-z_6,0,0,z_4-z_9,0,0)\nonumber
 \\
 (0,0,0,0,-z_6,-z_6,-z_4-z_6+z_9,s_{12}+s_{13}+s_{23}-z_4-z_6+z_9,0,s_{12}+s_{13}+s_{14}+s_{23}-z_4-z_6+z_9,s_{12}+s_{13}+s_{23}+s_{24}-z_4-z_6+z_9)\nonumber
 \\
(2
z_1,z_1,-s_{12}+z_1+z_3,-s_{12}-s_{13}-s_{23}+z_1+z_4,-s_{12}-s_{13}-s_{23}+z_1+z_4+z_6-z_8-z_9,z_1+z_6-z_8,0,0,z_1+z_9,0,0)\nonumber
\\
 (0,0,0,0,-s_{12}-s_{13}-s_{23}+z_1+z_4+z_6-z_8-z_9,z_1+z_6-z_8,z_6-z_8-z_9,-z_1+z_6-z_8,0,z_6-z_8,s_{12}-z_1-z_3+z_6-z_8)\nonumber
 \\
(-z_1+z_6-z_8,-z_1+z_6-z_{10},-z_1+z_6+z_8-z_{10}-z_{11},s_{12}+s_{13}+s_{23}-z_1-z_4+z_9,s_{12}+s_{13}+s_{23}-z_1-z_4+z_8+z_9,%
\nonumber \\-z_1+z_6+z_8,0,0,z_6-z_1,0,0)\nonumber
\\
(0,0,0,0,s_{12}+s_{13}+s_{23}-z_1-z_4+z_8+z_9,-z_1+z_6+z_8,z_8,2
z_8,0,z_8+z_{10},z_8+z_{11})\,,
\label{M1_257}
\end{gather}
}
whereas the generators of $\tilde M_2$ on the cut
$\{2,5,7\}$ take the form,
\begin{align}
 &(z_1,0,0,0,0,0,0,0,0,0,0), (0,0,z_3,0,0,0,0,0,0,0,0)\nonumber \\ & (0,0,0,z_4,0,0,0,0,0,0,0), (0,0,0,0,0,z_6,0,0,0,0,0)\nonumber \\ & (0,0,0,0,0,0,0,z_8,0,0,0), (0,0,0,0,0,0,0,0,1,0,0)\nonumber \\ & (0,0,0,0,0,0,0,0,0,1,0) , (0,0,0,0,0,0,0,0,0,0,1)\,.
 \label{M2_257}
\end{align}
The intersection of $\tilde M_1$ and $\tilde M_2$, with the
generators given above, is then computed by the method described
in section~\ref{section_intersection}. In the case at hand,
we proceed as follows.
\begin{enumerate}
\item First we compute the generators $\tilde M_1 \cap \tilde
M_2$ in the ring, with the algorithm described in Lemma \ref{lem int},
\begin{gather}
  \label{eq:1}
  \mathbb B = \mathbb Q[z_1, z_3, z_4, z_6, z_8, z_9, z_{10}, z_{11}, s_{12}, s_{13},
  s_{14}, s_{23}, s_{24}]
\end{gather}
using a block ordering with $[z_1, z_3, z_4, z_6, z_8, z_9, z_{10}, z_{11}] \succ [s_{12}, s_{13},
  s_{14}, s_{23}, s_{24}]$.
\item Then we map the generators of $\tilde M_1 \cap \tilde
M_2$ from the previous step, to the ring,
\begin{gather}
  \label{eq:1}
  \tilde{\mathbb A} = \mathbb Q( s_{12}, s_{13},
  s_{14}, s_{23}, s_{24}) [z_1, z_3, z_4, z_6, z_8, z_9, z_{10}, z_{11}]
\end{gather}
and simplify the generators.
\item Finally, we use the heuristic algorithm (Algorithm \ref{alg trim}) to delete redundant
  generators in $\tilde M_1 \cap \tilde
M_2$.
\end{enumerate}
These steps are automated by our {\sc Singular} program. The generators of $\tilde M_1 \cap \tilde M_2$ for the other $9$ cuts listed in eq.~\eqref{eq:spanning_set_of_cuts}
were obtained in the same manner. In table \ref{timings} we provide timings for the computation of the module
intersections for the relevant cuts of the  non-planar hexagon-box diagram. The timings are
in seconds on an Intel Xeon E5-2643 machine with $24$ cores, $3.40$ GHz
and $384$ GB
of RAM.

\begin{table}[ptb]
\begin{center}%
\begin{tabular}
[c]{|l|r|r|}%
\hline
cut & time / sec & mem/GB\\\hline
$\{1,5,7\}$ & $218$& 4.3\\
$\{2,5,7\}$ & $43$&1.1\\
$\{2,5,8\}$ & $303$&6.7\\
$\{2,6,7\}$ & $743$&9.8\\
$\{3,5,8\}$ & $404$&7.4\\
$\{3,6,7\}$ & $699$&11.0\\
$\{3,6,8\}$ & $24$&1.0\\
$\{4,6,8\}$ & $797$&13.7\\
$\{1,4,5,8\}$ & $53$&1.7\\
$\{1,4,6,7\}$ & $196$&3.0\\
\hline
\end{tabular}
\end{center}
\caption{Timings and RAM usages for the module intersection computations for the relevant cuts
of the non-planar hexagon-box diagram.}%
\label{timings}
\end{table}

Furthermore, using the heuristics as specified in  Algorithm \ref{alg trim} we are able to
reduce the size of the generating systems as specified in table
\ref{tab trim}. For instance,
the trimmed generating system for the module intersection on the cut $\{2,5,7\}$
consists of $24$ generators, which are fully analytic in
$s_{12}, s_{13}, s_{14}, s_{23}, s_{24}$. The generators
can be downloaded from
\begin{quote}
  \url{https://raw.githubusercontent.com/yzhphy/hexagonbox_reduction/master/cut257/module_intersection_257.txt}\,.
\end{quote}
Each list in this file is in the format
\begin{equation}
  \label{eq:11}
  (\tilde b_1,0,\tilde b_3,\tilde b_4,0,\tilde b_6,0,\tilde b_8,\tilde
  a_9,\tilde a_{10},\tilde a_{11}, \tilde b) \,.
\end{equation}

\begin{table}[ptb]
\begin{center}%
\begin{tabular}
[c]{|c|r|r|}%
\hline
cut & original size / MB & trimmed size / MB \\\hline
$\{1,5,7\}$ & $68$ & $10$\\
$\{2,5,7\}$ & $25$ & $1.4$\\
$\{2,5,8\}$ & $49$ & $3.1$\\
$\{2,6,7\}$ & $100$ & $2.8$\\
$\{3,5,8\}$ & $97$ & $3.7$\\
$\{3,6,7\}$ & $80$ & $3.6$\\
$\{3,6,8\}$ & $10$ & $1.6$\\
$\{4,6,8\}$ & $21$ & $1.6$\\
$\{1,4,5,8\}$ & $4.4$ & $3.6$\\
$\{1,4,6,7\}$ & $9.4$ & $4.1$\\
\hline
\end{tabular}
\end{center}
\caption{String sizes of the original and trimmed generating systems, given in megabytes.}%
\label{tab trim}
\end{table}

\subsection{Reduction of IBP identities on cuts}
After computing the module intersections on the $10$ cuts, we use
eq.~\eqref{IBP_cut} to generate IBP identities without doubled propagators
on each cut in turn. As described in section~\ref{section_REF}, we use linear algebra techniques to select the
relevant and independent IBP identities for reducing the target integrals to the
master integrals on each cut. Characteristics of the resulting linearly
independent linear systems, all analytic in $s_{12}, s_{13}, s_{14}, s_{23}, s_{24}$
and the spacetime dimension $D$, are presented in table~\ref{tab:linear_systems}.
\begin{table}[ht]
  \centering
  \begin{tabular}{|c|c|c|c|c|}
\hline
    cut & \# equations & \# integrals & byte size / MB & density \\
\hline
$\{1,5,7\}$ & $1144$ & $1177$ & $1.2$ & $1.4\%$\\
$\{2,5,7\}$ & $1170$ & $1210$ & $0.99$ & $1.3\%$\\
$\{2,5,8\}$ & $1152$ & $1190$ & $1.1$ &$1.5\%$\\
$\{2,6,7\}$& $1118$& $1155$ & $1.0$ & $1.5\%$\\
$\{3,5,8\}$&$1160$& $1202$& $1.2$ &$1.5\%$\\
$\{3,6,7\}$&$1173$& $1217$ & $1.3$ & $1.7\%$\\
$\{3,6,8\}$&$1135$& $1176$ & $0.77$ & $1.2\%$\\
$\{4,6,8\}$&$1140$& $1176$ & $0.94$  & $1.2\%$\\
$\{1,4,5,8\}$&$700$& $723$ & $0.69$ & $1.7\%$\\
$\{1,4,6,7\}$&$683$& $706$ & $0.66$ & $1.6\%$\\
\hline
  \end{tabular}
  \caption{Linear systems of IBP identities on the $10$ cuts given
  in eq.~\eqref{eq:spanning_set_of_cuts} for reducing the target hexagon-box
    integrals in eq.~\eqref{eq:target}. The storage size is for the
    corresponding matrix and is measured in megabytes. The density refers to the
    percentage of non-vanishing elements in the matrices of these linear systems.}
  \label{tab:linear_systems}
\end{table}
It is clear from this table that these linear systems are very
sparse and of relatively small byte size.

For the purpose of demonstration, we have made the IBP relations on the cut $\{2,5,7\}$ available at
\begin{quote}
  \url{https://raw.githubusercontent.com/yzhphy/hexagonbox_reduction/master/cut257/hexagonbox_257_deg4.txt}
  \, .
\end{quote}

We calculated the row reduced echelon form (RREF) of these linear systems
using the {\it total pivoting} strategy, in order to retain the sparsity in
intermediate steps. For the linear systems corresponding to some cuts,
we are able to directly obtain the RREF analytically
in $s_{12}, s_{13}, s_{14}, s_{23}, s_{24}$ and the spacetime dimension $D$.
For linear systems corresponding to other cuts, we calculate the RREF with integer values
of one or two $s_{ij}$ repeatedly. The fully analytic RREF is then readily
obtained by our private heuristic interpolation algorithm. These
computations are implemented in our primitive {\sc Mathematica}
code. A much more efficient RREF code in {\sc Singular} is in
preparation.

We also remark that for the computation of different cuts, a useful
trick is to work with different choices of independent Mandelstam
invariants. For a specific cut, a suitable choice of Mandelstam
invariants can speed up the computation and also save the RAM
usage. After the RREF is obtained, we use the program {\sc Fermat}
\cite{Fermat} to
replace the new Mandelstam variables by the original choice $s_{12},
s_{13}, s_{14}, s_{23}, s_{24}$.

The running time and resources required of our {\sc Mathematica} code
depends on the size of the linear systems and also {\it the complexity of coefficients} in
the reduced IBPs. For the smallest linear system, the one corresponding to the
quadruple cut $\{1,4,6,7\}$, our {\sc Mathematica} RREF code obtained the
fully analytical RREF
in $31$ minutes with one core and $1.5$ GB RAM usage on a
laptop with $16$ GB RAM. For the largest linear system, that
corresponding to the triple cut $\{3,6,7\}$, we run our {\sc Mathematica}
code and assign integer values to two Mandelstam invariants. This finished in $2.5$ hours
and used $1.8$ GB RAM. We parallelize the running with
various integer values, evaluating $440$ points on the IRIDIS High Performance Computing Facility. The
semi-analytic results are then
interpolated to the fully analytic RREF result by our heuristic multivariate
interpolation algorithm, requiring a CPU time of $23$ minutes with one core
and $15$ GB RAM usage.

\subsection{Merging of IBP reductions and final result}
Having obtained the RREF of the $10$ IBP systems on cuts, it is
straightforward to merge the coefficients to obtain the complete IBP reductions
without applied cuts \cite{Larsen:2015ped}. For example, to determine the
reduction coefficient of $\mathcal{I}_{42} =I(1, 0, 1, 0, 1, 1, 1, 0, 0, 0, 0)$,
we search for $\mathcal{I}_{42}$ in table~\ref{table:pre-master} and find that it is supported on the cuts
$\{1,5,7\}$ and $\{3,6,7\}$. We find that, for every target integral, the reduction
coefficient for $\mathcal I_{42}$ on the cut $\{1,5,7\}$ equals the
corresponding coefficient on the cut $\{3,6,7\}$, as must be the case.
Thus, for each target integral in turn, we obtain its reduction coefficient
of each $\mathcal{I}_j$ as the reduction coefficient of $\mathcal{I}_j$ on each cut
supporting $\mathcal{I}_j$. In this way we obtain the complete reduction
of the target integral \emph{without applied cuts}.

After merging the $10$ RREFs we reduce the $32$ target integrals in
eq.~\eqref{eq:target} to the $75$ ``pre''-master integrals
$\{ \mathcal{I}_1, \ldots, \mathcal{I}_{75} \}$. We then apply
the global symmetry in eq.~\eqref{global_symmetry} to eliminate the redundant
integrals $\mathcal{I}_{63}$ and $\mathcal{I}_{68}$, and then finally obtain the
complete analytic reduction to the $73$ master integrals. The IBP
reduction result, as replacement rules (in a compressed file of the
size $280$ MB), can be downloaded from the link provided in the introduction.

\subsection{Comparison with other IBP solvers}

We have checked our results with FIRE5~\cite{Smirnov:2014hma} in the C++ implementation with
LiteRed~\cite{Lee:2012cn}, and KIRA (v 1.1)
\cite{Maierhoefer:2017hyi}. We note that it is not an easy task to
perform the analytic IBP reduction of
hexagon-box integrals with degree-four numerators. Due
to the RAM limit we faced,  
we have not yet been able to obtain the analytic IBP reduction for
degree-four-numerator hexagon-box integrals from FIRE5 or KIRA with
the Rackham cluster, on a node with two $10$ core Intel Xeon V4 CPU and $256$ GB of
RAM\footnote{We are currently running the analytic IBP reduction with
FIRE5 and Kira on a node with more RAM available, and presently waiting for the result.}. On the
other hand, it is easy to obtain the numeric IBP reduction for
degree-four-numerator hexagon-box integrals with FIRE5 and KIRA. For
example, if all the five $s_{ij}$ are taken to be integers, FIRE5 is
able to generate the purely numeric IBP reductions in about $6.0$ hours. We ran FIRE5
purely numerically for many sets of integer values, and all the numeric IBP reduction results are
consistent with our analytic IBP reductions, after a basis change between
the {\sc Azurite} integral basis and the FIRE basis.

\section{Conclusion and outlook}\label{section_conclusion}
In this paper we have presented a new and efficient method for
computing integration-by-parts (IBP) reductions based on the ideas developed
in refs.~\cite{Larsen:2015ped, Zhang:2016kfo}.
We used a module intersection method \cite{Zhang:2016kfo} to trim IBP
systems on unitarity cuts. The key idea is the efficient analytic
computation of module intersections, which is achieved by the
mathematical technique of treating the kinematical
parameters as variables and using a block monomial ordering for which
[variables] $\succ$ [parameters]. This trick could also be helpful for
other types of multi-loop multi-scale amplitude computations. After solving the module intersection
problems, we find linear IBP systems on cuts which are of very small sizes. This part
is implemented in our highly efficient and automated {\sc Singular}
code. For example, the analytic IBP vectors for the hexagon-box integral with
no doubled propagators on triple cuts can be computed in minutes.

Furthermore, we applied sophisticated sparse linear algebra techniques to
compute the row reduced echelon form of the IBP identities on cuts. For
example, we applied a weighted version of the Markowitz pivoting
strategy to retain the sparsity in intermediate steps of the row
reduction. We have implemented the sparse linear algebra part of our
algorithm as a preliminary {\sc Mathematica} code. A more
efficient {\sc Singular} implementation will become available in the near future.

In this paper, we have solved a cutting-edge IBP reduction problem fully analytically:
that of the reduction of non-planar five-point hexagon-box integrals,
with numerators of degree four, to the basis of $73$
master integrals. Our result has been verified (numerically) with the
state-of-art IBP reduction programs.

We are currently preparing an automated implementation based on open-source
software, such as {\sc Singular}, with which we expect to be able to
solve yet more difficult IBP reduction problems. We expect that our method will
boost the computation of NNLO cross sections for more
complicated $2\rightarrow 3$ scattering processes and higher-multiplicity cases.

Finally, we also expect that the ideas presented in this paper, such as the
module intersection for trimming integral relations, the special ordering for
variables and parameters, and sparse linear algebra techniques, can
also be combined in various ways with other existing computational methods such as:
the finite field sampling and reconstruction approach
\cite{vonManteuffel:2014ixa, Peraro:2016wsq}, the dual conformal
symmetry construction of IBP vectors \cite{Bern:2017gdk} and
the newly developed D-module method \cite{Bitoun:2017nre}, in order to
determine tailored optimal reduction strategies for various classes of
Feynman integrals.

\section*{Acknowledgments}
We thank Babis Anastasiou, Simon Badger, Zvi Bern, Rutger Boels, Christian Bogner, Jorrit Bosma,
Charles Bouillaguet,  Wolfram Decker, Lance Dixon, James Drummond, {\"O}mer G{\"u}rdo{\u g}an,
Johannes\- Henn, Enrico Herrmann, Harald Ita, Mikhail Kalmykov, David Kosower, Dirk Kreimer, Roman N. Lee, Hui Luo, Andreas von
Manteuffel, Alexander Mitov, David Mond, Erik Panzer, Costas
Papadopoulos, Tiziano Peraro, Gerhard Pfister,
Robert Schabinger, Peter Uwer, Gang Yang and Mao
Zeng for very enlightening discussions.
The research leading to these results has received
funding from Swiss National Science Foundation (Ambizione grant PZ00P2
161341), from the National Science Foundation under Grant No. NSF
PHY17-48958, and from the European Research Council (ERC) under
the European Union's Horizon 2020 research and innovation programme (grant agreement No 725110).
The work of YZ is also partially supported by the
Swiss National Science Foundation through the NCCR SwissMap, Grant number 141869.
The work of AG is supported by the Knut and Alice Wallenberg Foundation under grant \#2015-0083.
The work of KJL is supported by ERC-2014-CoG, Grant number 648630 IQFT.
The work of JB and HS was supported by Project II.5 of SFB-TRR 195
``Symbolic Tools in Mathematics and their Application'' of the German Research Foundation (DFG).
The authors acknowledge the use of the IRIDIS High Performance Computing Facility,
and associated support services at the University of Southampton.
Part of the computations were also performed on resources provided
by the Swedish National Infrastructure for Computing (SNIC) at Uppmax.
The authors moreover acknowledge the use of the Euler computing cluster,
associated with ETH Z{\"u}rich.

\bibliography{hb_reduction}
\end{document}